\documentclass[hidelinks,journal]{IEEEtran}

\usepackage{amsmath,amssymb,amsfonts}
\usepackage{array}
\PassOptionsToPackage{table}{xcolor}
\usepackage{textcomp}
\usepackage{url}
\usepackage{verbatim}
\usepackage{graphicx}
\usepackage{cite}
\usepackage[dvipsnames]{xcolor}
\definecolor{blue}{rgb}{0,0,0}\def\NewFigOpacity{0}
\usepackage{orcidlink}    
\usepackage[acronym,shortcuts]{glossaries}
\usepackage{bm}
\usepackage[font=footnotesize]{caption}
\usepackage{subcaption}
\usepackage{relsize}
\usepackage{soul}
\usepackage{algorithm, algpseudocode}
\usepackage{algcompatible}
\usepackage{lipsum}
\usepackage{mathtools}
\usepackage[bottom]{footmisc}
\usepackage{tabularx}
\usepackage{stfloats}
\usepackage{amsthm}
\usepackage{makecell}
\usepackage{threeparttable}

\def\BibTeX{{\rm B\kern-.05em{\sc i\kern-.025em b}\kern-.08em
T\kern-.1667em\lower.7ex\hbox{E}\kern-.125emX}}

\def\thesubsubsectiondis{\unskip\arabic{subsubsection})}

\newcommand{\herm}[0]{^{\mathsf{H}}}

\newcommand{\Eric}[1]{{\color{blue}{#1}}}
\usepackage{tikz}
\setbox0\hbox{\Eric{\xglobal\colorlet{EricCol}{.}}}
\providecommand{\NewFigOpacity}{0.10}
\newcommand{\NewFig}[1]{\begin{tikzpicture}\node[inner sep=0pt,outer sep=0pt](nf){#1};\fill[blue,fill opacity=\NewFigOpacity](nf.south west) rectangle (nf.north east);\end{tikzpicture}}
\makeatletter
\newcommand{\NewRef}[2]{\@namedef{CPnewref@#2}{#1}}
\def\CP@refcol{}
\let\CP@thebibliography\thebibliography
\def\thebibliography#1{\CP@thebibliography{#1}\let\CP@bibitem\bibitem\def\bibitem##1{\CP@setrefcol{##1}\CP@bibitem{##1}}}
\def\CP@setrefcol#1{\edef\CP@want{\ifcsname CPnewref@#1\endcsname\csname CPnewref@#1\endcsname\fi}%
  \ifx\CP@want\CP@refcol\else\ifx\CP@want\@empty\normalcolor\else\color{\CP@want}\fi\let\CP@refcol\CP@want\fi}
\makeatother
\NewRef{blue}{wen2017multiple}\NewRef{blue}{jian2024fda}\NewRef{blue}{rou2026evolving}
\NewRef{blue}{mirabella2026continuous}\NewRef{blue}{liu2025cp}\NewRef{blue}{rou2026normalized}
\NewRef{EricCol}{yu2022location}\NewRef{EricCol}{wang2024ris}\NewRef{EricCol}{rou2026resurrection}
\NewRef{EricCol}{niu2025afdm}\NewRef{EricCol}{marodin2026physical}
\NewRef{blue}{sui2026achieving}\NewRef{blue}{sui2026generalized}

\newtheorem{theorem}{Theorem}
\newtheorem{lemma}{Lemma}
\newtheorem{corollary}{Corollary}

\newacronym{OFDM}{OFDM}{orthogonal frequency division multiplexing}
\newacronym{DFT-s-OFDM}{DFT-s-OFDM}{DFT-spread OFDM}
\newacronym{FBMC}{FBMC}{filter bank multi-carrier}
\newacronym{GFDM}{GFDM}{generalized frequency division multiplexing}
\newacronym{OCDM}{OCDM}{orthogonal chirp division multiplexing}
\newacronym{AFDM}{AFDM}{affine frequency division multiplexing}
\newacronym{OTFS}{OTFS}{orthogonal time frequency space}
\newacronym{T-OTFS}{T-OTFS}{transcendentally-rotated OTFS}
\newacronym{ODDM}{ODDM}{orthogonal delay-Doppler division multiplexing}
\newacronym{SWHM}{SWHM}{sparse Walsh-Hadamard multiplexing}
\newacronym{FMCW}{FMCW}{frequency modulated continuous wave}
\newacronym{OTSM}{OTSM}{orthogonal time sequency modulation}
\newacronym{MMSE}{MMSE}{minimum mean squared error}
\newacronym{ICI}{ICI}{inter-carrier interference}
\newacronym{RMSE}{RMSE}{root-mean-squared error}
\newacronym{CPP}{CPP}{chirp-periodic prefix}
\newacronym{PDA}{PDA}{probabilistic data association}

\newacronym{LCT}{LCT}{linear canonical transform}
\newacronym{ZT}{ZT}{Zak transform}
\newacronym{DZT}{DZT}{discrete Zak transform}
\newacronym{IDZT}{IDZT}{inverse discrete Zak transform}
\newacronym{FT}{FT}{Fourier transform}
\newacronym{IFT}{IFT}{inverse Fourier transform}
\newacronym{DFT}{DFT}{discrete Fourier transform}
\newacronym{IDFT}{IDFT}{inverse discrete Fourier transform}
\newacronym{AFT}{AFT}{affine Fourier transform}
\newacronym{DAFT}{DAFT}{discrete affine Fourier transform}
\newacronym{IDAFT}{IDAFT}{inverse discrete affine Fourier transform}
\newacronym{SFFT}{SFFT}{symplectic finite Fourier transform}
\newacronym{ITN}{ITN}{intelligent traffic network}
\newacronym{ISFFT}{ISFFT}{inverse symplectic finite Fourier transform}
\newacronym{HT}{HT}{Heisenberg transform}
\newacronym{SISO}{SISO}{single-input single-output}
\newacronym{MIMO}{MIMO}{multiple-input multiple-output}
\newacronym{WT}{WT}{Wigner transform}
\newacronym{frFT}{frFT}{fractional Fourier transform}
\newacronym{IfrFT}{IfrFT}{inverse fractional Fourier transform}
\newacronym{fnT}{fnT}{Fresnel transform}
\newacronym{IfnT}{IfnT}{inverse Fresnel transform}
\newacronym{LT}{LT}{Laplace transform}
\newacronym{ILT}{ILT}{inverse Laplace transform}
\newacronym{ISAC}{ISAC}{integrated sensing and communications}
\newacronym{JCAS}{JCAS}{joint communications and sensing}
\newacronym{EM}{EM}{electromagnetic}
\newacronym{CP}{CP}{chirp-permuted}
\newacronym{B5G}{B5G}{beyond fifth generation}
\newacronym{3GPP}{3GPP}{$3^{\rm{rd}}$ generation partnership project}           
\newacronym{4G}{4G}{fourth generation}                                          
\newacronym{5G}{5G}{fifth generation}                                           
\newacronym{6G}{6G}{sixth generation}
\newacronym{CPIM}{CPIM}{chirp-permutation index modulation}
\newacronym{LTV}{LTV}{linear time-variant}
\newacronym{LTI}{LTI}{linear time-invariant}
\newacronym{LTVM}{LTVM}{linear time-variant multipath}
\newacronym{TV}{TV}{time-variant}
\newacronym{TI}{TI}{time-invariant}
\newacronym{1D}{1D}{one-dimensional}
\newacronym{2D}{2D}{two-dimensional}
\newacronym{3D}{3D}{three-dimensional}
\newacronym{NTN}{NTN}{non-terrestrial network}
\newacronym{LEO}{LEO}{low earth orbit}
\newacronym{IoT}{IoT}{Internet-of-Things}
\newacronym{mmWave}{mmWave}{millimeter-wave}
\newacronym{THz}{THz}{Terahertz}
\newacronym{V2X}{V2X}{vehicle-to-everything}
\newacronym{RCC}{RCC}{radar-communication coexistence}
\newacronym{C-V2X}{C-V2X}{Cellular-V2X}                                             
\newacronym{NR}{NR}{New Radio}                  
\newacronym{Tbps}{Tbps}{Terabits-per-second}
\newacronym{ms}{ms}{millisecond}
\newacronym{UAV}{UAV}{unmanned aerial vehicle}
\newacronym{CFO}{CFO}{carrier frequency offset}
\newacronym{ITS}{ITS}{intelligent transportation system}
\newacronym{SAGIN}{SAGIN}{space-air-ground integrated network}
\newacronym{DTN}{DTN}{digital twin network}                                    
\newacronym{ETSI}{ETSI}{European Telecommunications Standards Institute}            
\newacronym{EHF}{EHF}{extremely high-frequency}
\newacronym{BER}{BER}{bit-error-rate}
\newacronym{SotA}{SotA}{state-of-the-art}
\newacronym{DoF}{DoF}{degrees-of-freedom}
\newacronym{UE}{UE}{user equipment}
\newacronym{AP}{AP}{access point}
\newacronym{CIR}{CIR}{channel impulse response}
\newacronym{CSI}{CSI}{channel state information}
\newacronym{SNR}{SNR}{signal-to-noise ratio}
\newacronym{AWGN}{AWGN}{additive white Gaussian noise}
\newacronym{ML}{ML}{maximum likelihood}
\newacronym{OOB}{OOB}{out-of-band}
\newacronym{MX}{MX}{multiplexing}
\newacronym{DMX}{DMX}{demultiplexing}
\newacronym{AoA}{AoA}{angle-of-arrival}
\newacronym{AoD}{AoD}{angle-of-departure}
\newacronym{5GAA}{5GAA}{5G Automotive Association}                                  
\newacronym{6G-IA}{6G-IA}{6G Smart Networks and Services Industry Association}      
\newacronym{MUSIC}{MUSIC}{MUltiple SIgnal Classification}
\newacronym{ESPRIT}{ESPRIT}{estimation of signal parameters via rotational invariance techniques}
\newacronym{DP}{DP}{detection problem}
\newacronym{EP}{EP}{estimation problem}
\newacronym{KPI}{KPI}{key performance indicator}
\newacronym{ISI}{ISI}{inter-symbol interference}
\newacronym{TVIRF}{TVIRF}{time-variant impulse response function}
\newacronym{TVTF}{TVTF}{time-variant transfer function}
\newacronym{DVIRF}{DVIRF}{Doppler-variant impulse response function}
\newacronym{DVTF}{DVTF}{Doppler-variant transfer function}
\newacronym{IM}{IM}{index modulation}
\newacronym{DDSF}{DDSF}{delay-Doppler spread function}
\newacronym{PAPR}{PAPR}{peak-to-average power ratio}
\newacronym{AFBM}{AFBM}{affine filter bank modulation}
\newacronym{OOBE}{OOBE}{out-of-band emission}
\newacronym{RCS}{RCS}{radar cross-section}
\newacronym{PIR}{PIR}{peak interference residual}
\newacronym{PSR}{PSR}{peak-to-sidelobe ratio}
\newacronym{PA}{PA}{power amplifier}
\newacronym{RF}{RF}{radio frequency}
\newacronym{JCDE}{JCDE}{joint channel and data estimation}
\newacronym{AF}{AF}{ambiguity function}
\newacronym{FFT}{FFT}{fast Fourier transform}
\newacronym{RIS}{RIS}{reconfigurable intelligent surface}
\newacronym{LFM}{LFM}{linear frequency modulation}

\newacronym{i.i.d.}{i.i.d.}{independent and identically distributed}
\newacronym{PSLR}{PSLR}{peak-to-sidelobe ratio}
\newacronym{ISLR}{ISLR}{integrated sidelobe ratio}
\newacronym{CDF}{CDF}{cumulative distribution function}
\newacronym{CCDF}{CCDF}{complementary cumulative distribution function}
\newacronym{CP-AFDM}{CP-AFDM}{chirp-permuted affine frequency division multiplexing}
\newacronym{QAM}{QAM}{quadrature amplitude modulation}

\begin{document}
\bstctlcite{IEEEexample:BSTcontrol}





\title{Chirp-Permuted AFDM: \\ \Eric{A Versatile and Multifunctional} Waveform for 6G}

\author{Hyeon Seok Rou\textsuperscript{\orcidlink{0000-0003-3483-7629}}\!,~\IEEEmembership{Member,~IEEE},
{\color{blue}Hanvit Kim\textsuperscript{\orcidlink{0000-0002-2486-3316}}\!,~\IEEEmembership{Graduate Student Member,~IEEE},} \\
{\color{blue}Sunwoo Kim\textsuperscript{\orcidlink{0000-0002-7055-6587}}\!,~\IEEEmembership{Senior Member,~IEEE},} and
Giuseppe Thadeu Freitas de Abreu\textsuperscript{\orcidlink{0000-0002-5018-8174}}\!,~\IEEEmembership{Senior Member,~IEEE}.
\thanks{\Eric{H.~S.~Rou is with the School of Electrical Engineering, Korea Advanced Institute of Science and Technology (KAIST), Daejeon, 34141, South Korea (email: erichsrou@kaist.ac.kr)}}
\thanks{\color{blue} H.~Kim and S.~Kim are with the Department of Electronics and Computer Engineering, Hanyang University, Seoul, 04763, South Korea (emails: [dante0813, remero]@hanyang.ac.kr). }
\thanks{\Eric{G.~T.~F.~de~Abreu is with the School of Computer Science and Engineering, Constructor University, Campus Ring 1, 28759 Bremen, Germany (email: gabreu@constructor.university).}}
\vspace{-4ex}}

\markboth{Submitted to the IEEE transactions on wireless communications}%
{H. S. Rou \MakeLowercase{\textit{et al.}}}


\maketitle

\begin{abstract}
We present a novel multicarrier waveform, termed \ac{CP-AFDM}, which introduces a unique chirp-permutation domain on top of the chirp subcarriers of the conventional \acs{AFDM}. 
\textcolor{blue}{Rigorous analysis of the signal model and waveform properties demonstrates that the proposed \ac{CP-AFDM} waveform preserves all core characteristics of \ac{AFDM} -- including robustness to doubly-dispersive channels, \ac{PAPR}, and \Eric{\ac{ISAC}-enabling} full delay-Doppler representation -- while introducing a new degree-of-freedom without any \Eric{loss of these properties}.}
\textcolor{blue}{Given this, by exploiting the vast chirp-permutation domain, two selected exemplary multifunctional applications are presented: an \ac{IM} technique over the permutation domain which achieves significant spectral efficiency gains, and a physical-layer security scheme that ensures practically perfect security through permutation-based keying, without requiring additional transmit energy or signaling overhead.}
In all, these characteristics establish \ac{CP-AFDM} as a highly attractive candidate for emerging \ac{6G} use cases demanding reliability, sensing-awareness, \Eric{and  multifunctionality.}
\end{abstract}

\vspace{-1ex}

\begin{IEEEkeywords}
\Acf{AFDM}, doubly-dispersive channel, chirp permutation, waveform design, multifunctional, \acs{ISAC}, \acf{6G}.
\end{IEEEkeywords}

\glsresetall

\IEEEpeerreviewmaketitle

\vspace{-1ex}

\section{Introduction}
\label{sec:introduction}

The evolution toward \ac{6G} systems, to be adopted by standardization by 2030, is envisioned to unlock unprecedented capabilities in wireless communication technology~\cite{Wang_6G}, with key performance criteria including ultra-high data rate, ultra-low latency, and massive connectivity, extending orders of magnitude beyond what is outlined in the current \ac{5G}~\cite{agiwal2016next}.
These benchmarks are outlined to support emerging scenarios defined by the IMT-2030 vision, encompassing \ac{ISAC}, seamless connectivity, immersive experiences, in addition to a broad range of new applications including \ac{IoT}, \ac{V2X}, \ac{SAGIN}, \ac{DTN}, and more \cite{saad2019vision}.

Achieving such demanding requirements necessitates pushing the wireless operation into higher spectral bands, namely, \ac{mmWave} and \ac{THz} frequencies, where sensitivity of the channel to phase-shift effects becomes particularly severe.
In addition, a great number of applications and scenarios inherently include high-mobility and heterogeneous environments, further subjecting the wireless channel to severe Doppler shifts, \Eric{even including advanced metasurface applications such as multiple \ac{RIS} deployment in a high-mobility environment \cite{yu2022location,wang2024ris}.}
These consolidate the need for considering doubly-dispersive wireless channels for \ac{6G}, characterized by both time- and frequency-selectivity~\cite{Bliss_DDchannel}, which lead to considerable challenges for traditional waveform designs, such as \ac{OFDM} which suffers from pronounced \ac{ISI} and \ac{ICI} under high-mobility and multipath due to degradation of subcarrier orthogonality~\cite{armada2001understanding}. 
%
%

In light of these challenges, waveform design for \ac{6G} has emerged as a critical research frontier~\cite{liyanaarachchi2021optimized,rou2026resurrection}, with recent efforts exploring alternative signal representations that remain resilient under dispersive channel conditions.
\Ac{OTFS}, for example, modulates information over a two-dimensional delay-Doppler grid via a cascade of the \ac{ISFFT} and a Heisenberg transform, or equivalently through the \ac{IDZT}, effectively converting rapidly time-varying multipath channels into approximately invariant delay-Doppler-domain channels~\cite{hadani2018otfs,gopalam2024zak}.
This structural transformation provides robust orthogonality over doubly-dispersive channels and hence improved communications performance in high-Doppler environments, and enhanced parameter estimation for mobility-aware applications. 
Consequently, \ac{OTFS} has been successfully extended to support advanced \ac{MIMO} and \ac{ISAC} frameworks~\cite{gaudio2020effectiveness,ranasinghe2024fast}.

Another promising candidate is \ac{AFDM} \cite{Bemani_AFDM21,Bemani_AFDM23}, which generalizes the conventional \ac{FT} using the \ac{AFT} parameterized by two tunable chirp rates. 
This formulation unifies and extends existing waveforms, subsuming both \ac{OFDM} and \ac{OCDM} \cite{ouyang2016orthogonal} as special cases, while enabling the modulation kernel to be tailored to the delay and Doppler characteristics of the channel \cite{Rou_SPM24,Bemani_AFDM23}. 
As a result, \ac{AFDM} ensures full diversity over doubly-dispersive environments and achieves full delay-Doppler orthogonality and representation, thereby also motivating several successful extensions in various contexts~\cite{yin2024diagonally,RanasingheWCNC,zhu2023design,ranasinghe2025affine}.
In addition, \ac{AFDM} has demonstrated great advantages for practical deployment in terms of standardization potential, thanks to its superior backward compatibility with the current system, i.e., \ac{OFDM}, hardware efficiency, and flexible parameterization \cite{rou2025affine,savaux2024special}.

Building upon the advantages and design principles of the \ac{AFDM} waveform, this article introduces \ac{CP-AFDM}, a novel waveform that extends \ac{AFDM} by incorporating controlled permutations in the chirp domain.
It should be noted that while \ac{CP-AFDM} has been previously introduced in preliminary form for selected application scenarios, such as chirp-permutation \ac{IM} \cite{rou2024afdm} and physical layer security \cite{rou2025chirp}, it has not yet been formally proposed, nor fundamentally analyzed in terms of its underlying modulator structure and waveform properties and delay-Doppler characteristics.
\textcolor{blue}{While the existing permutation-based waveform designs, such as multiple-mode \ac{OFDM}-index modulation~\cite{wen2017multiple} and the frequency offset permutation index modulation for \ac{ISAC}~\cite{jian2024fda}, permute the data or the physical resources onto which it is mapped, \Eric{the proposed \ac{CP}-\ac{AFDM} permutes only the chirp sequences in the chirp-domain.}}

%

In light of the above, we present in this article the following contributions:
\begin{itemize}
\item \textcolor{blue}{The \ac{CP-AFDM} waveform is formally proposed, consolidating a rigorous formulation of the chirp-permutation domain and the modulator structure based on the \ac{CP}-\ac{DAFT}, \Eric{which, for general permutations, lies outside the DAFT and \ac{LCT} families}}.
\item \textcolor{blue}{\Eric{The \ac{CP}-\ac{DAFT} and the effective channel structure} of the \ac{CP-AFDM} waveform is rigorously analyzed to show that the beneficial properties of the conventional \ac{AFDM} are completely preserved, by showing that the effective channel of the one-sided \ac{CP-AFDM} is a unitary similarity transformation of that of \ac{AFDM}.}
\item \textcolor{blue}{The delay-Doppler characteristics of \ac{CP-AFDM} are analyzed via the \ac{AF} for all-ones and random symbols, revealing that the chirp permutation \Eric{enhances the Doppler resolution without degrading the delay resolution} in the deterministic probe regime by equalizing the energy distribution of the transmit signal.}
\item \textcolor{blue}{The chirp-permutation domain is established as a new degree of freedom for waveform-level multifunctionality, attained\Eric{, for a given permutation,} without any additional \Eric{computational complexity}, signaling overhead, or transmit energy -- and is briefly exemplified by \Eric{two of our prior works:} an index modulation scheme, and a physical-layer security scheme based on permutation keying.}
\end{itemize}

The remainder of the article is organized as follows:
Section~\ref{sec:system_model} reviews the doubly-dispersive channel model and outlines the fundamentals of conventional \ac{AFDM};
Section~\ref{sec:proposed_CP-AFDM_waveform} presents the formal definition of \ac{CP-AFDM}, including its modulator structure and signal model;
Section~\ref{sec:analysis} provides a comprehensive analysis of \ac{CP-AFDM}, supported by theoretical derivations and empirical results, including \ac{BER} performance and ambiguity function characteristics;
Section~\ref{sec:application} explores application scenarios that benefit from the chirp-permutation domain, demonstrating how \ac{CP-AFDM} enables integrated multifunctionality;
and finally, Section~\ref{sec:conclusion} concludes the article with a summary and outlook on future research directions for this newly proposed waveform.

\section{System Model and Fundamentals}
\label{sec:system_model}

In this section, we provide the foundation for the proposed \ac{CP-AFDM} waveform design to be presented in Section~\ref{sec:proposed_CP-AFDM_waveform}, by briefly outlining the doubly-dispersive wireless channel model and the signal model of the conventional \ac{AFDM} waveform.

\vspace{-0.5ex}
\subsection{Doubly-Dispersive Channel}
\label{subsec:dd_channel}

We consider a general wireless communication scenario in which the channel between a transmitter and receiver consists of $P$ significant multipath components. 
Each $p$-th path, indexed by $p \in \{1,\cdots,P\}$, is characterized by a complex fading coefficient $h_p \in \mathbb{C}$, a delay $\tau_p \in [0, \tau^\mathrm{max}]$, and a Doppler shift $\nu_p \in [-\nu^\mathrm{max}, +\nu^\mathrm{max}]$, where $\tau^\mathrm{max}$ and $\nu^\mathrm{max}$ denote the maximum path delay and Doppler spread of the channel, respectively.

The above leads to a classical \ac{LTV} input-output model, commonly described using the \ac{TVIRF} given by \vspace{-0.5ex}
\begin{equation}
h(t, \tau) \triangleq \sum_{p=1}^{P} h_p \cdot \Eric{e^{-j2\pi \nu_p t}} \cdot \delta(\tau - \tau_p),
\label{eq:TVIR} \vspace{-0.5ex}
\end{equation}
where $j \triangleq \sqrt{-1}$, $t$ is the continuous time index, $\tau$ is the path delay, and $\delta(\cdot)$ denotes the Dirac delta function.

The model of \eqref{eq:TVIR} assumes a finite number of delay and Doppler taps, which is a well-established approximation of practical doubly-dispersive wireless channels, particularly in underspread conditions~\cite{Bliss_DDchannel,Rou_SPM24}, which implies that the maximum delay spread $\tau_\mathrm{max} - \tau_\mathrm{min}$ is smaller than $T$, the maximum Doppler spread $\nu_\mathrm{max} - \nu_\mathrm{min}$ is smaller than $\frac{1}{T}$, and $\tau_\mathrm{max} \nu_\mathrm{max} <\!< 1$, where $T$ is the finite signal period in seconds.

Then, given a baseband transmit signal $s(t)$ of bandwidth $B$, the received signal under this channel model is typically expressed as a convolution over the delay axis. 
By sampling the continuous signal at a sufficient sampling frequency of $f_\mathrm{s}$, the discrete input-output relation is given by \vspace{-0.5ex}
\begin{gather}
r[n] = {\sum_{\ell = 0}^{\infty}}  s[n - \ell] \left( \sum_{p = 1}^{P} h_p \cdot \Eric{e^{-j2\pi f_p \frac{n}{N}}} \cdot {\delta \big[ \ell - \ell_p \big]} \right) + w[n], \nonumber \\[-2ex]
\label{eq:IO_discrete}
\end{gather}
\vspace{-4.5ex}

\noindent where $r[n]$, $s[n]$, and $w[n]$ denote the sampled discrete sequences of the receive, transmit, and noise signals, respectively; $f_p \triangleq \frac{N \nu_p}{f_\mathrm{s}} \in [-\frac{N\nu_\mathrm{max}}{f_\mathrm{s}}, \frac{N\nu_\mathrm{max}}{f_\mathrm{s}}]$ is the normalized digital Doppler shift of the $p$-th path; $\ell_p \triangleq \frac{\tau_p}{T_\mathrm{s}} \in \{0,\cdots\!,\ell_{\mathrm{max}}\}$ is the normalized integer delay\footnote{It is assumed that each path delay is accurately quantized to an integer multiple of the sampling period $T_\mathrm{s}$, i.e., $\ell_p = \tau_p / T_\mathrm{s}$ with negligible rounding error, as commonly assumed in the literature.} of the $p$-th path; and $\delta[\,\cdot\,]$ is the discrete unit impulse function.
{\color{blue}Note that the specified integer delay-based model can be generalized into the fractional delay-based model, which is presented as~\cite{rou2026evolving}
\begin{gather}
r[n]=\sum_{p=1}^P h_p \Eric{e^{-j 2 \pi f_p \frac{n}{N}}} \sum_{m=0}^{N-1} s[m] g\left((n-m)-\ell_p\right)+w[n], \nonumber \\[-2ex]
\label{eq:IO_discrete_Fractional}
\end{gather}}
\vspace{-4ex}

\noindent where $g(\cdot)$ denotes the effective discrete-time pulse kernel.
{\color{blue}In this study, we consider the integer delay model\footnote{\textcolor{blue}{The analysis of the fractional delay-based model to the \ac{AFDM} can be found in~\cite{mirabella2026continuous}, where its effect would be \Eric{equivalent for the one-sided} \ac{CP-AFDM}\Eric{, whose effective channel is unitarily similar to that of} the \ac{AFDM}.}} in \eqref{eq:IO_discrete}.
}

Following the above, modern multicarrier systems appended a prefix sequence to mitigate \ac{ISI} from time dispersion, which is defined over $N_\mathrm{cp}$ samples, where $N_\mathrm{cp} \geq \ell_\mathrm{max}$, such that
\begin{equation} 
s[n'] = s[N + n'] \cdot e^{j2\pi \cdot \phi_\mathrm{cp}(n')},
\label{eq:cyclic_prefix}
\end{equation}
with $n' \in \{-1,\cdots,-N_\mathrm{cp}\}$, and $\phi_\mathrm{cp}(n')$ denoting a waveform-dependent phase term (e.g., $\phi_\mathrm{cp}(n') = 0$ for \ac{OFDM} and \ac{OTFS}, or chirp-periodic for \ac{AFDM}).

After discarding the prefix at the receiver, the convolution becomes circular and can be expressed in matrix form as
\begin{align}
\mathbf{r} \triangleq \mathbf{H} \!\cdot\! \mathbf{s} & = \Bigg( \sum_{p=1}^{P} \underbrace{h_p \!\cdot\! \mathbf{\Phi}_{p} \!\cdot\! \mathbf{W}^{f_p} \!\cdot\! \mathbf{L}^{\ell_p}}_{\triangleq \mathbf{H}_p \in \mathbb{C}^{N \times N}} \Bigg)\!\cdot\! \mathbf{s} + \mathbf{w},
\label{eq:IO_matrix}
\end{align}
where $\mathbf{r} \!\in\! \mathbb{C}^{N \times 1}$, $\mathbf{s} \!\in\! \mathbb{C}^{N \times 1}$, and $\mathbf{w} \!\in\! \mathbb{C}^{N \times 1}$ are respectively the vectors representing the received signal, the transmit signal, and \ac{AWGN}; and the channel matrix $\mathbf{H}$ is composed of $P$ components $\mathbf{H}_p$, each parameterized by three deterministic matrices:
\begin{itemize}
\item $\mathbf{\Phi}_p \in \mathbb{C}^{N \times N}$, a diagonal matrix of prefix-dependent phase offsets $\phi_\mathrm{cp}(n)$, given by \eqref{eq:CP_phase_matrix},
\item $\mathbf{W} \in \mathbb{C}^{N \times N}$, a diagonal matrix of the $N$-th roots of unity, given by \eqref{eq:W_matrix},
\item $\mathbf{L} \in \mathbb{C}^{N \times N}$, a forward cyclic shift matrix, given by \vspace{-1ex}
\end{itemize}

\begin{figure*}[t]
\begin{equation}
\mathbf{\Phi}_{p} \triangleq \mathrm{diag}\Big(\big[\overbrace{e^{-j2\pi\cdot \phi_\mathrm{cp}(\ell_p)}, e^{-j2\pi\cdot \phi_\mathrm{cp}(\ell_p - 1)}, \cdots, e^{-j2\pi\cdot \phi_\mathrm{cp}(2)}, e^{-j2\pi\cdot \phi_\mathrm{cp}(1)}}^{\ell_p \;\text{terms}}, \overbrace{\;\!1\;\!, 1\;\!, \cdots\!\vphantom{e^{(x)}}, 1\;\!, 1}^{N - \ell_p\;\! \text{ones}}\big]\Big) \in \mathbb{C}^{N \times N}\!.
\label{eq:CP_phase_matrix}
\end{equation}
\begin{equation}
\mathbf{W} \triangleq \mathrm{diag}\Big(\big[1, e^{-j2\pi/N}, \cdots, e^{-j2\pi(N-2)/N}, e^{-j2\pi(N-1)/N}\big]\Big) \in \mathbb{C}^{N \times N}.
\label{eq:W_matrix}
\end{equation}
\hrule
\vspace{-1ex}
\end{figure*}

\begin{equation}
\label{eq:forwardcyclic _matrix}
\mathbf{L} =
\begin{bmatrix}
0      & 0      & \cdots & 0      & 1      \\
1      & 0      & \cdots & 0      & 0      \\
0      & 1      & \ddots & \vdots & \vdots \\
\vdots & \ddots & \ddots & 0      & 0      \\
0      & \cdots & 0      & 1      & 0
\end{bmatrix} \in \mathbb{C}^{N \times N},
\end{equation}
such that right-multiplying by $\mathbf{L}^{\ell_p}$ corresponds to a cyclic left-shift of a matrix by $\ell_p \in \mathbb{N}_0$ elements\footnote{\Eric{Again, a generalization of the delay-dependent cyclic shift matrix to support non-integer delays has been investigated in the literature~\cite{rou2026evolving}. 
However, as previously assumed, we focus on the integer-delay case, and the consequential effects of the fractional delays have been studied in~\cite{mirabella2026continuous}.}}.

In light of the above, the convolutional matrix structure implies that Doppler shifts are not directly resolvable unless an additional domain transform is applied, as different paths with identical delays yet different Dopplers, will completely overlap in the diagonal position.
Therefore, a key objective of advanced waveform design is to achieve delay-Doppler orthogonality in channel structure, via appropriate transformations.

\subsection{The Conventional AFDM Waveform}
\label{subsec:AFDMwaveform}

The conventional \ac{AFDM} waveform \cite{Bemani_AFDM21,Bemani_AFDM23} is based on the \acf{DAFT} to map the transmit signal into a time-frequency domain representation, whose $N$-point forward transform matrix is defined as

$~$
\vspace{-1ex}
\begin{equation}
\mathbf{A} \triangleq \mathbf{\Lambda}_{c_2} \mathbf{F}_N \mathbf{\Lambda}_{c_1} \in \mathbb{C}^{N \times N},
\end{equation}
where $ \mathbf{F}_N \in \mathbb{C}^{N \times N}$ is the $N$-point \ac{DFT} matrix, and $\mathbf{\Lambda}_{c_i} \triangleq \mathrm{diag}[e^{-j2\pi c_i (0)^2}, \cdots, e^{-j2\pi c_i (N-1)^2}] \in \mathbb{C}^{N \times N}$ is a diagonal chirp matrix defined by a central digital frequency $c_i$, with two parametrizable chirp parameters $c_1$ and $c_2$.

Given the above, the inverse transform, i.e., \ac{IDAFT} is utilized to modulate the data vector $\mathbf{x} \in \mathbb{C}^{N \times 1}$ onto the chirp-domain subcarriers, as
\begin{equation}
\mathbf{s}^{\mathrm{AFDM}} \triangleq \mathbf{A}^{\!-1} \cdot \mathbf{x} = \overbrace{(\mathbf{\Lambda}_{c_1}\herm \cdot \mathbf{F}_N\herm \cdot \mathbf{\Lambda}_{c_2}\herm )}^{\text{IDAFT}} \cdot \; \mathbf{x} \in \mathbb{C}^{N \times 1}.
\label{eq:tx_afdm}
\end{equation} 

Then, the received \ac{AFDM} signal $\mathbf{r}^\mathrm{AFDM} \in \mathbb{C}^{N \times 1}$ after propagation through the doubly-dispersive channel in \eqref{eq:IO_matrix} is demodulated by the matched \ac{DAFT} matrix $\mathbf{A} \in \mathbb{C}^{N\times N}$ of the same parametrization on $c_1$ and $c_2$, yielding the demodulated signal vector $\mathbf{y}^{\mathrm{AFDM}} \in \mathbb{C}^{N \times 1}$ as
\begin{align}
\mathbf{y}^{\mathrm{AFDM}}&  = \mathbf{A} \cdot \overbrace{(\mathbf{H} \cdot \mathbf{s}^{\mathrm{AFDM}} + {\mathbf{w}})}^{\triangleq \mathbf{r}^\mathrm{AFDM} \; \in \; \mathbb{C}^{N \times 1} } \nonumber \\
& = \mathbf{G}^{\mathrm{AFDM}}\cdot\mathbf{x} +  \mathbf{A}\mathbf{w} \in \mathbb{C}^{N \times 1}.
\label{eq:rx_afdm}
\end{align} 

The effective \ac{AFDM} channel of \eqref{eq:rx_afdm} is defined by
\begin{align}
\mathbf{G}^{\mathrm{AFDM}} \in \mathbb{C}^{N \times N}&  \label{eq:eff_afdm}\\[-0.5ex] 
& \hspace{-16ex} \triangleq \sum_{p=1}^{P}  h_p \!\cdot\! \big(\mathbf{\Lambda}_{c_2} \!\cdot\! \mathbf{F}_N \!\cdot\! \mathbf{\Lambda}_{c_1} \big)  \Big (\mathbf{\Phi}_{p} \!\cdot\! \mathbf{W}^{f_p} \!\cdot\! \Eric{\mathbf{L}^{\ell_p}}  \Big) \big(\mathbf{\Lambda}_{c_1}\herm \!\cdot\! \mathbf{F}_N\herm \!\cdot\! \mathbf{\Lambda}_{c_2}\herm \big), \nonumber
\end{align} 
where the prefix phase offset function of \ac{AFDM} in $\mathbf{\Phi}_{p}$ is given by $\phi^\mathrm{AFDM}_{\mathrm{cp}}(n) = c_1(N^2 + 2Nn)$, incorporating a \ac{CPP} as required by \ac{AFDM} systems\footnote{For even values of $N$ and integer values of $2Nc_1$, the phase offset matrix of \ac{AFDM} reduces to the identity matrix $\mathbf{I}_N$.}.

Given the above, the \ac{AFDM} effective channel achieves full orthogonality over integer delay-Doppler indices when the channel satisfies the orthogonality condition given by
\begin{equation}
2\big(f^\mathrm{max} + \xi\big)(\ell^\mathrm{max} + 1) + \ell^\mathrm{max} \leq N,
\label{eq:AFDM_orthoCondition}
\end{equation}
where $f^\mathrm{max}$ and $\ell^\mathrm{max}$ are the maximum unambiguous normalized digital Doppler shift and normalized integer delay of the doubly-dispersive channel, respectively, and the parameter $\xi \in \mathbb{N}_0$ denotes the so-called \textit{guard width}, indicating the number of additional elements reserved to account for Doppler interference leakage \cite{Rou_SPM24,Bemani_AFDM23}.

Provided that the orthogonality condition is satisfied, to ensure full diversity, the optimal chirp frequency $c_1$ for the first chirp must be set to $c_1 = \frac{2(f^\mathrm{max} + \xi) + 1}{2N}$, while the second chirp frequency $c_2$ may be chosen as an arbitrary irrational number, i.e., $c_2 \in \{ \mathbb{R} \setminus \mathbb{Q} \}$, or a sufficiently small real value such that $c_2 \ll \frac{1}{2N}$.

In other words, to achieve optimal communications performance over doubly-dispersive channels, while there is a strict requirement on the first chirp sequence and $c_1$, also referred to as the post-chirp domain, there is no strict requirement for $c_2$ beyond avoiding resonance with the Doppler harmonics of the doubly-dispersive channel.
Indeed, leveraging this inherent flexibility in $c_2$, not available in any other waveforms, a plethora of \ac{AFDM} extensions have been proposed, including \ac{PAPR} reduction techniques \cite{sui2026achieving,ali2025spreading}, \ac{IM} schemes \cite{sui2026generalized,zhu2023design,liu2024pre}, and ambiguity function tuning \cite{yin2025ambiguity,bedeer2025ambiguity}.

In the following section, we further extend this degree of freedom in a novel direction, leveraging a chirp-permutation operation, leading to the proposed \ac{CP}-\ac{AFDM} waveform.

\section{Proposed Chirp-Permuted AFDM (CP-AFDM)}
\label{sec:proposed_CP-AFDM_waveform}

\subsection{Permutation Operation}
\label{subsec:permutation}

A permutation operation can be described in terms of a bijective mapping from the set of $N$ indices $\{0, 1, \ldots, N-1\}$ onto itself, which can be denoted by the index mapping function $\Pi_i(\,\cdot\,) : \{0,\ldots,N-1\} \rightarrow \{0,\ldots,N-1\}$, where $i \in \{1,\ldots,N!\}$ denotes specific permutation order from the total of $N!$ possible permutations. 
Given an arbitrary column vector $\mathbf{x} \in \mathbb{C}^{N \times 1}$, the $i$-th permutation of its elements yielding the permuted sequence $\mathbf{x}_i \in \mathbb{C}^{N \times 1}$ is described as
\begin{equation}
\mathbf{x}_i \triangleq 
\begin{bmatrix}
{x_{\Pi_i(0)}} \\
\vdots \\
x_{\Pi_i(n)} \\
\vdots \\
{x_{\Pi_i(N-1)}}\\[0.5ex]
\end{bmatrix}
= \mathbf{\Pi}_i \cdot\!
\begin{bmatrix}
{x_0} \\
\vdots \\
x_n \\
\vdots \\
{x_{N-1}}\\[0.5ex]
\end{bmatrix}
= \mathbf{\Pi}_i \mathbf{x} \in \mathbb{C}^{N \times 1},
\label{eq:permutation}
\end{equation}
where $[\mathbf{x}_i]_n \triangleq x_{\Pi_i(n)}$ is the new element at position $n$ in the permuted sequence.

\subsection{Chirp-Permuted Discrete Affine Fourier Transform}
\label{subsec:chirp-permuted_AFT}

In light of the above, let us define the \ac{CP}-\ac{DAFT} as the \ac{DAFT} with \Eric{permuted chirp sequences}.

The general $N$-point \ac{CP}-\ac{DAFT} is defined as
\begin{align}
\mathbf{A}_{i_1,i_2} &\triangleq \mathrm{diag}\big(\mathbf{\Pi}_{i_2} \!\cdot\! \boldsymbol{\lambda}_{c_2}\big) \cdot \mathbf{F}_N \cdot \mathrm{diag}\big(\mathbf{\Pi}_{i_1} \!\cdot\! \boldsymbol{\lambda}_{c_1}\big) \nonumber \\
& =  \mathbf{\Lambda}_{c_2,i_2} \mathbf{F}_N \mathbf{\Lambda}_{c_1,i_1} \in \mathbb{C}^{N \times N},
\label{eq:cpdaft_matrix}
\end{align}
in other words, the \ac{CP}-\ac{DAFT} is obtained by permuting the chirp sequences $\boldsymbol{\lambda}_{c_1},\boldsymbol{\lambda}_{c_2}$ in the diagonal chirp matrices $\mathbf{\Lambda}_{c_1},\mathbf{\Lambda}_{c_2}$, and $i_1,i_2 \in \{1,\ldots,N!\}$ are the permutation orders of the first and second chirp sequences, respectively.

\textcolor{blue}{It is important to highlight that the CP-DAFT is \Eric{in general} not a re-parametrization of the DAFT, but a distinct transform.
Indeed, the DAFT constitutes a continuous family parametrized by the pair $(c_{1}, c_{2})$, \Eric{and for a general permutation, no choice of chirp parameters reproduces the permuted kernel}.
Therefore, the CP-DAFT can be regarded as a novel, purely-discrete transform, in likeness to the Hadamard transform~\cite{knagenhjelm2002hadamard} and the discrete cosine transform~\cite{ahmed2006discrete}.}


Furthermore, we show that the \ac{CP}-\ac{DAFT} is a unitary transform even under arbitrary chirp-permutations, and therefore retains the useful properties such as basis orthonormality, invertibility, and energy conservation, which will be leveraged in the design of a new waveform modulation.

\begin{lemma}[Unitarity of diagonal chirp matrices]
\label{lem:chirp_unitary}
Let $\boldsymbol{\lambda} = [e^{-j2\pi c (0)^2}\!, \ldots\!, e^{-j2\pi c (N-1)^2}]$ be a quadratic chirp sequence with central chirp frequency $c \in \mathbb{R}$, and define the diagonal matrix $\mathbf{\Lambda} = \mathrm{diag}(\boldsymbol{\lambda}) \in \mathbb{C}^{N \times N}$. Then $\mathbf{\Lambda}$ is unitary, i.e.,
\begin{equation}
\mathbf{\Lambda}\herm  \cdot\mathbf{\Lambda} = \mathbf{I}_N.
\end{equation}

\begin{proof}
By definition of the product of two diagonal matrices, we have $\left[ \mathbf{\Lambda}\herm \mathbf{\Lambda} \right]_{n,n'} = 0 ~\, \forall n \neq n'$, i.e., the off-diagonal elements are zero, while the diagonal elements are given by
\begin{equation*}
\left[ \mathbf{\Lambda}\herm \mathbf{\Lambda} \right]_{n,n} = \big(e^{-j2\pi c (n)^2}\big)^{\!*} \!\cdot \big(e^{-j2\pi c (n)^2}\big) = \big| e^{-j2\pi c (n)^2}\big| = 1,
\end{equation*}
for all $n \in \{0,\ldots,N-1\}$.
\end{proof}

\end{lemma}

\begin{corollary}[Unitarity of permuted chirp matrices]
\label{cor:permuted_unitary}
Let $\Pi_i(\cdot)$ be a permutation operator with given order $i$ as defined in Section \ref{subsec:permutation}, such that the diagonal permuted-chirp matrix is defined as $\mathbf{\Lambda}_i = \mathrm{diag}(\mathbf{\Pi}_i\boldsymbol{\lambda}) = \mathrm{diag}\big([e^{-j2\pi c \cdot \Pi_i(0)^2}, \ldots\!, e^{-j2\pi c \cdot\Pi_i(N-1)^2}]\big) \in \mathbb{C}^{N \times N}$.
Then, $\mathbf{\Lambda}_{i}$ is unitary for all $i \in \{1,\ldots,N!\}$, i.e.,
\begin{equation}
\mathbf{\Lambda}_i\herm \cdot \mathbf{\Lambda}_i^{\vphantom{H}} = \mathbf{I}_N, ~ \forall i \in \{1,\ldots,N!\}
\end{equation}
from Lemma \ref{lem:chirp_unitary} and $\big(e^{-j2\pi c \cdot \Pi_i(n)^2}\big)^{\!*} \big(e^{-j2\pi c \cdot \Pi_i(n)^2}\big) = 1$.
\end{corollary}

%
%
\begin{theorem}[Unitarity of the CP-DAFT Matrix]
\label{thm:cpdaft_unitary}
Let $\mathbf{F}_N$ denote the $N$-point \ac{DFT} matrix, and let $\mathbf{\Lambda}_{c_1,i_1} = \mathrm{diag}(\mathbf{\Pi}_{i_1}\boldsymbol{\lambda}_{c_1})$ and $\mathbf{\Lambda}_{c_2, i_2} = \mathrm{diag}(\mathbf{\Pi}_{i_2} \boldsymbol{\lambda}_{c_2})$ denote diagonal matrices whose entries are the permuted chirp sequences, respectively with central frequencies $c_1$ and $c_2$ and permutation orders $i_1$ and $i_2$.
Then, the \ac{CP}-\ac{DAFT} matrix defined as
\begin{equation}
\mathbf{A}_{i_1,i_2} \triangleq \mathbf{\Lambda}_{c_2, i_2} \cdot \mathbf{F}_N \cdot \mathbf{\Lambda}_{c_1,i_1}
\end{equation}
is unitary, i.e.,
\begin{equation}
\mathbf{A}_{i_1,i_2}\herm \cdot \mathbf{A}_{i_1,i_2}  = \mathbf{I}_N,
\end{equation}
for all $i_1,i_2 \in \{1,\ldots,N!\}$.
\end{theorem}
\begin{proof}
The proof follows from the fact that the \ac{DFT} matrix $\mathbf{F}_N$ is unitary, and that the diagonal chirp matrices $\mathbf{\Lambda}_{c_2, i_2}, \mathbf{\Lambda}_{c_1,i_1}$ are also unitary as shown in Corollary \ref{cor:permuted_unitary}.
\end{proof}

In general, the following condition with two \ac{CP}-\ac{DAFT} matrices with non-matching permutation orders, $i'_1 \neq i_1$, $i'_2 \neq i_2$, should be highlighted for clarification: 
\begin{equation}
\mathbf{A}_{i'_1,i'_2}\herm \cdot \mathbf{A}_{i_1,i_2}^{\vphantom{H}}  \neq \mathbf{I}_N.
\end{equation}

\subsection{Signal Model of Chirp-Permuted {AFDM} (CP-AFDM)}
\label{sec:signal_CPAFDM}

Given the consolidated \ac{CP}-\ac{DAFT}, we propose a novel waveform modulation scheme, which we refer to as the \ac{CP-AFDM} waveform.
Specifically, the time-domain \ac{CP-AFDM} signal $\mathbf{s} \in \mathbb{C}^{N \times 1}$ is obtained by applying the inverse \ac{CP}-\ac{DAFT} with permutation parameters $i_1,i_2$ to the transmit symbols $\mathbf{x} \in \mathcal{X}^{N \times 1}$, which is given by
\begin{equation}
\label{eq:txsig_cpafdm}
\mathbf{s}_{i_1,i_2} \triangleq \mathbf{A}_{i_1,i_2}^{-1}\mathbf{x} = \mathbf{\Lambda}_{c_1,i_1}\herm \mathbf{F}_N\herm \mathbf{\Lambda}_{c_2,i_2}\herm \mathbf{x} \in \mathbb{C}^{N \times 1}.
\end{equation}

Or alternatively, the modulated \ac{CP-AFDM} signal can be represented in terms of a linear transform kernel, similar to that of other multicarrier waveforms, including \ac{OFDM} and \ac{AFDM}, where $\kappa_n(m)$ is the modulation kernel 
\begin{align}
\kappa_n(m) &\triangleq \tfrac{1}{\sqrt{N}}e^{j2\pi \cdot c_1 \cdot \Pi_{i_1}\!(n)^2}\!\! \cdot e^{j2\pi \frac{mn}{N}} \cdot e^{j2\pi c_2 \cdot \Pi_{i_2}\!(m)^2}, \label{eq:kernal_CPAFDM}  \\
& = \tfrac{1}{\sqrt{N}} \mathrm{exp}\Big[j2\pi \big( c_1 \Pi_{i_1}\!(n)^2 + c_2 \Pi_{i_2}\!(m)^2 + \tfrac{mn}{N} \big) \Big], \nonumber
\end{align}
such that each $n$-th sample of the modulated \ac{CP-AFDM} signal can be expressed concisely as
\begin{equation}
[\mathbf{s}_{i_1,i_2}]_n = \sum_{m=0}^{N-1} x_m \cdot \kappa_n(m) \in \mathbb{C}.
\end{equation}

Trivially, when $i_1 = i_2 = 1$, i.e., no permutation is applied to either of the chirp sequences, the \ac{CP-AFDM} reduces to the conventional \ac{AFDM} waveform, as described in \eqref{eq:tx_afdm}.

\begin{figure*}
\centering
\includegraphics[width=1.775\columnwidth]{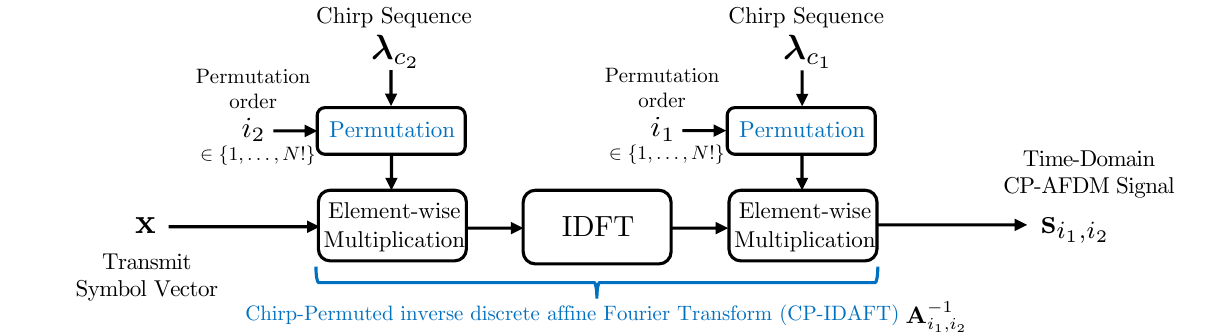}
\caption{A schematic block diagram of the proposed \ac{CP-AFDM} transmitter structures, as described by \eqref{eq:txsig_cpafdm}, highlighting the two additional permutation operations in the \ac{CP-AFDM} transmitter onto the chirp sequences.}
\label{fig:superfigure}
\vspace{-1ex}
\end{figure*}

\begin{figure*}[b!]
\setcounter{equation}{29}
\vspace{-2ex}
\hrule
\begin{equation}
\mathbf{G}_{i_1,i_2} \triangleq \mathbf{A}_{i_1,i_2}^{\vphantom{-1}}  \cdot \mathbf{H} \cdot \mathbf{A}_{i_1,i_2}^{-1} = \big(\mathbf{\Lambda}_{c_2,i_2} \cdot \mathbf{F}_N \cdot \mathbf{\Lambda}_{c_1,i_1}\big) \cdot \Big( \sum_{p=1}^{P} h_p \!\cdot\! \mathbf{\Phi}_{p} \!\cdot\! \mathbf{W}^{f_p} \!\cdot\! \mathbf{L}^{\ell_p}\Big) \cdot \big(\mathbf{\Lambda}_{c_1,i_1}\herm \cdot \mathbf{F}_N\herm \cdot \mathbf{\Lambda}_{c_2,i_2}\herm\big) \in \mathbb{C}^{N \times N}.\label{eq:CPAFDM-effective_channel}
\end{equation}
\vspace{-1.5ex}
\begin{equation}
\mathbf{G}_{i_2} \triangleq \mathbf{A}_{i_2}^{\vphantom{-1}}  \cdot \mathbf{H} \cdot \mathbf{A}_{i_2}^{-1} = \big(\mathbf{\Lambda}_{c_2,i_2} \cdot \mathbf{F}_N \cdot \mathbf{\Lambda}_{c_1}\big) \cdot \Big( \sum_{p=1}^{P} h_p \!\cdot\! \mathbf{\Phi}_{p} \!\cdot\! \mathbf{W}^{f_p} \!\cdot\! \mathbf{L}^{\ell_p}\Big) \cdot \big(\mathbf{\Lambda}_{c_1}\herm \cdot \mathbf{F}_N\herm \cdot \mathbf{\Lambda}_{c_2,i_2}\herm\big) \in \mathbb{C}^{N \times N}.
\label{eq:osCPAFDM-effective_channel}
\end{equation}
\vspace{-0.5ex}
\setcounter{equation}{23}
\end{figure*}

Then, the received signal $\mathbf{r}_{i_1,i_2} \in \mathbb{C}^{N \times 1}$ after the \ac{CP-AFDM} modulation over the channel in \eqref{eq:IO_matrix} is given by
\begin{equation}
\mathbf{r}_{i_1,i_2} = \mathbf{H} \cdot \mathbf{s}_{i_1,i_2} + \mathbf{w} \in \mathbb{C}^{N \times 1},
\end{equation}
which in turn is demodulated using the forward \ac{CP}-\ac{DAFT} with matching permutation parameters $i_1,i_2$ to yield the demodulated symbols $\mathbf{y}_{i_1,i_2} \in \mathbb{C}^{N \times 1}$,
\begin{equation}
\mathbf{y}_{i_1,i_2} = \mathbf{A}_{i_1,i_2} \cdot \mathbf{r}_{i_1,i_2} = \mathbf{A}_{i_1,i_2}^{\vphantom{-1}} \cdot \mathbf{H} \cdot \mathbf{A}_{i_1,i_2}^{-1} \cdot \mathbf{x} + \tilde{\mathbf{w}}_{i_1,i_2} \in \mathbb{C}^{N \times 1},
\end{equation}
with the transformed \ac{AWGN} vector $\tilde{\mathbf{w}}_{i_1,i_2} \triangleq \mathbf{A}_{i_1,i_2} \cdot \mathbf{w} \in \mathbb{C}^{N \times 1}$, which retains the same statistical distribution as $\mathbf{w}$, i.e., $\tilde{\mathbf{w}}_{i_1,i_2} \sim \mathcal{CN}(0,\sigma^2_w\mathbf{I}_N)$ as the \ac{CP}-\ac{DAFT} is unitary.

It is important to highlight that when non-matching permutation indices are utilized for demodulation, i.e., the received signal $\mathbf{r}_{i_1,i_2}$ is demodulated using the forward \ac{CP}-\ac{DAFT} $\mathbf{A}_{i'_1,i'_2}$ with permutation indices $i'_1 \neq i_1$ or $i'_2 \neq i_2$, the demodulated symbols will in general be significantly distorted yielding a non-informative signal.
This interesting property, unique to the proposed \ac{CP-AFDM}, can be exploited to enable efficient physical layer security, as will be elaborated in Section \ref{sec:application_secure}.

\subsection{One-Sided Chirp-Permuted AFDM Waveform}

Finally, let us also define a special case of the general \ac{CP}-\ac{DAFT} and consequently the \ac{CP-AFDM} waveform presented in \eqref{eq:cpdaft_matrix} and \eqref{eq:txsig_cpafdm}, where the permutation operation is only applied to the second chirp sequence and the first chirp sequence remains unpermuted, i.e., $i_1 = 1$ and $i_2 \in \{1,\ldots,N!\}$.
This subcase is dubbed the one-sided \ac{CP-AFDM} waveform, in order to differentiate from the previously defined general \textit{two-sided} \ac{CP-AFDM} waveform, where both chirp sequences are permuted.
As will be shown in the latter sections, this one-sided \ac{CP-AFDM} waveform is the useful version of the \ac{CP-AFDM} waveform, which retains all of the beneficial properties and characteristics of the conventional \ac{AFDM} in doubly-dispersive wireless channels, while still providing a new degree of freedom in the design of the transmit signal in the permutation index $i_2 \in \{1,\ldots,N!\}$.
In addition, as the one-sided \ac{CP-AFDM} is a subcase of the general \ac{CP-AFDM}, it also satisfies the unitarity and orthonormality conditions as analyzed in Section \ref{subsec:chirp-permuted_AFT}.
\textcolor{blue}{Furthermore, identical synchronization and prefixing procedures of \ac{AFDM} can be equivalently applied to one-sided \ac{CP-AFDM} since the prefix phase term depends only on the chirp parameter $c_{1}$, where its chirp-sequence $\lambda_{c_1}$ remains unpermuted in the one-sided \ac{CP-AFDM}.}

Finally, for completeness, the signal model for the one-sided \ac{CP-AFDM} is presented, with the $N$-point one-sided \ac{CP}-\ac{DAFT} matrix given by
\begin{align}
\mathbf{A}_{i_2} &\triangleq \mathrm{diag}\big(\mathbf{\Pi}_{i_2} \!\cdot\! \boldsymbol{\lambda}_{c_2}\big) \cdot \mathbf{F}_N \cdot \mathrm{diag}\big(\boldsymbol{\lambda}_{c_1}\big) \nonumber \\
& =  \mathbf{\Lambda}_{c_2,i_2} \mathbf{F}_N \mathbf{\Lambda}_{c_1} \in \mathbb{C}^{N \times N}.
\label{eq:cpdaft_onesided_matrix}
\end{align}

The corresponding transmit signal $\mathbf{s}_{i_2} \in \mathbb{C}^{N \times 1}$, received signal $\mathbf{r}_{i_2} \in \mathbb{C}^{N \times 1}$, and demodulated signal $y_{i_2} \in \mathbb{C}^{N \times 1}$ of the one-sided \ac{CP-AFDM} are given by
\begin{equation}
\mathbf{s}_{i_2} \triangleq \mathbf{A}_{i_2}^{-1}\mathbf{x} = \mathbf{\Lambda}_{c_1}\herm \mathbf{F}_N\herm \mathbf{\Lambda}_{c_2,i_2}\herm \mathbf{x} \in \mathbb{C}^{N \times 1},
\end{equation}
\begin{equation}
\mathbf{r}_{i_2} = \mathbf{H} \cdot \mathbf{s}_{i_2} + \mathbf{w} \in \mathbb{C}^{N \times 1},
\end{equation}
\begin{equation}
\mathbf{y}_{i_2} = \mathbf{A}_{i_2} \cdot \mathbf{r}_{i_2} = \mathbf{A}_{i_2}^{\vphantom{-1}} \cdot \mathbf{H} \cdot \mathbf{A}_{i_2}^{-1} \cdot \mathbf{x} + \tilde{\mathbf{w}}_{i_2} \in \mathbb{C}^{N \times 1}.
\end{equation}


\setcounter{equation}{31}
\section{Waveform and Performance Analysis}
\label{sec:analysis}

In this section, we analyze the proposed \ac{CP-AFDM} waveform and its one-sided subcase, in terms of the waveform characteristics and performance.
Specifically, it will be shown that the one-sided \ac{CP-AFDM} waveform retains many beneficial properties of the conventional \ac{AFDM} such as robustness and orthogonality over doubly-dispersive channels (and hence good communications performance), in addition to high delay-Doppler resolution, while also not altering any fundamental waveform properties such as \ac{PAPR}.
More interestingly, it will be shown that the \ac{CP-AFDM} even improves upon the \ac{SotA} waveforms in terms of the Doppler resolution and ambiguity function \acf{PSLR}. 

\vspace{-1ex}
\subsection{Effective Channel Analysis}
\label{subsec:effective_channel_analysis}

\begin{figure*}[b!]
    \vspace{-2.5ex}
    \centering
    \begin{subfigure}[b]{0.325\textwidth}
        \centering
        \includegraphics[width=\linewidth]{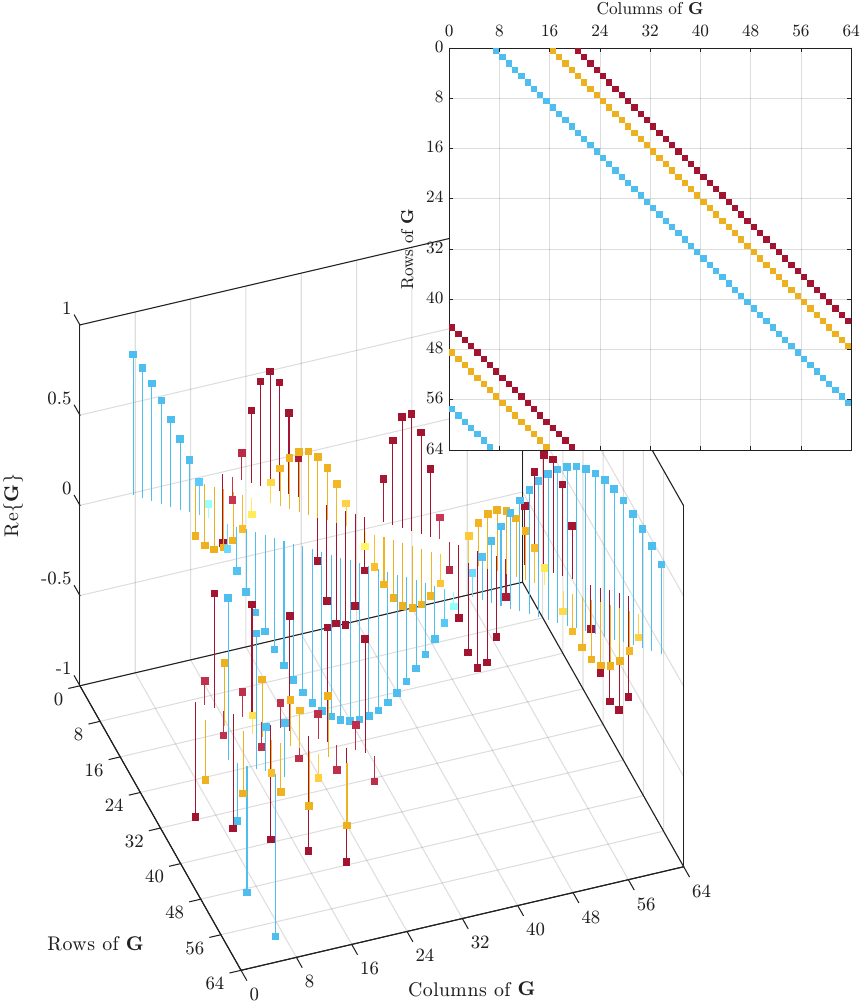}
                \vspace{-3ex}
        \caption{Conventional AFDM.}
        \label{fig:eff_channel_AFDM}
    \end{subfigure}
    \hfill
    \begin{subfigure}[b]{0.325\textwidth}
        \centering
        \includegraphics[width=\linewidth]{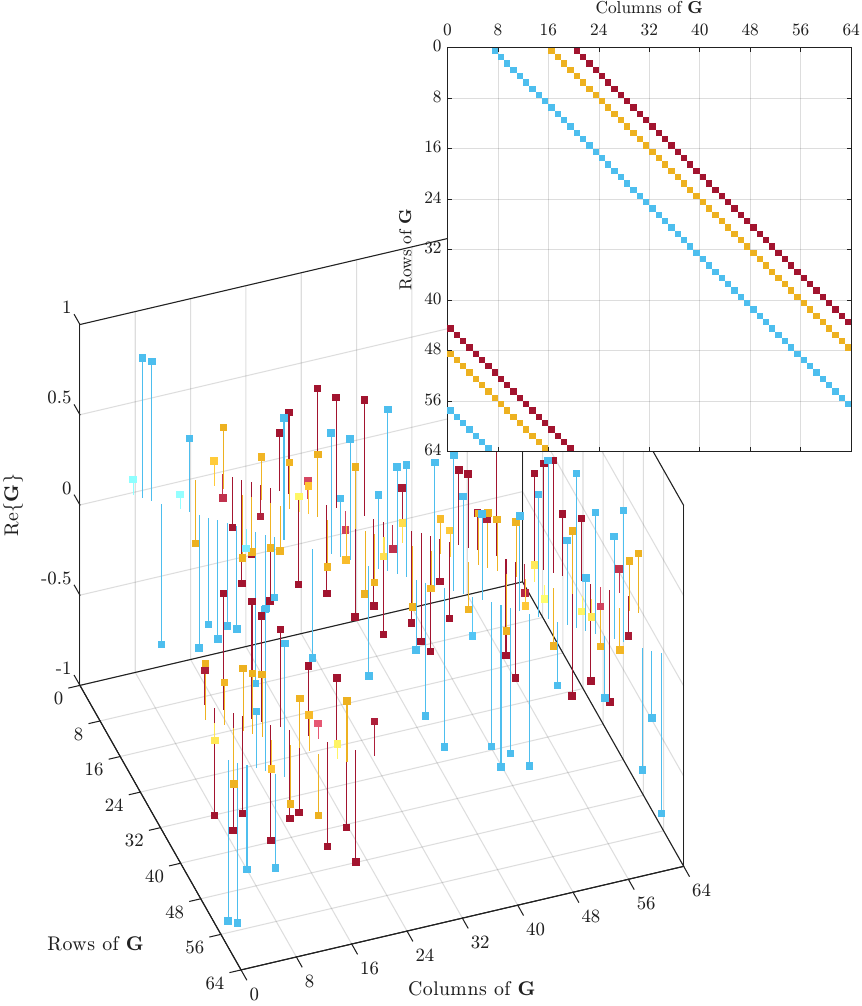}
                \vspace{-3ex}
        \caption{One-sided CP-AFDM.}
        \label{fig:eff_channel_3d}
    \end{subfigure}
    \hfill
    \begin{subfigure}[b]{0.325\textwidth}
        \centering
        \includegraphics[width=\linewidth]{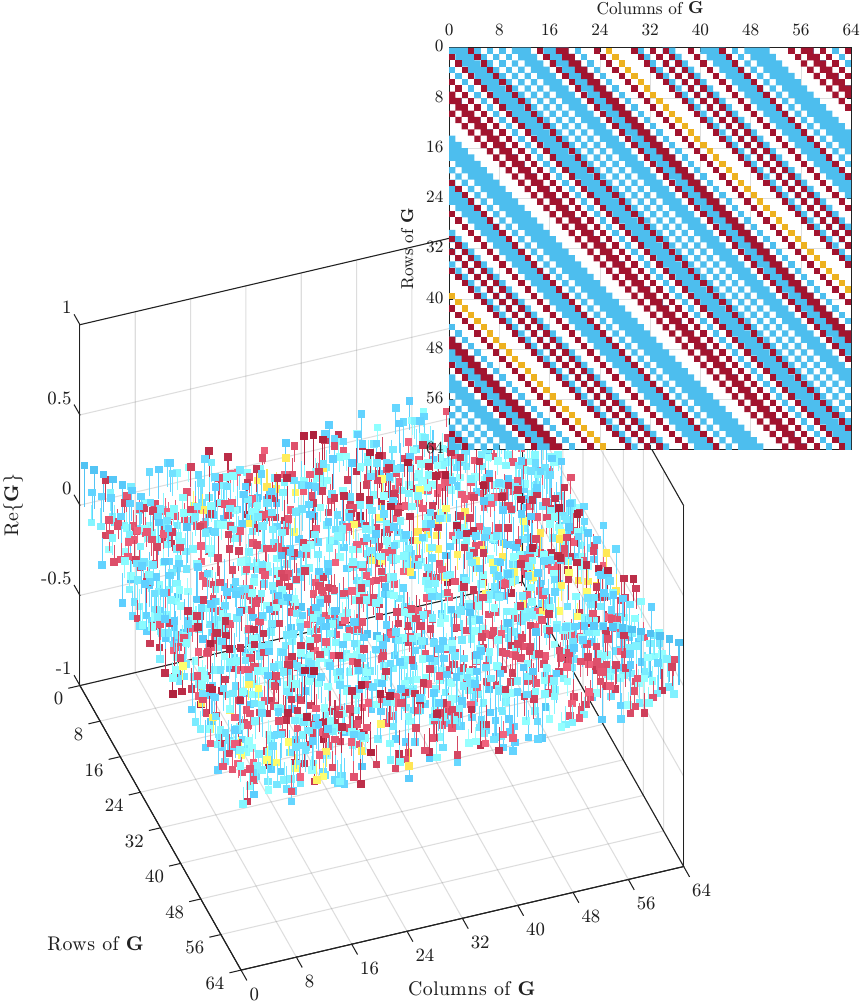}
        \vspace{-3ex}
        \caption{Two-sided CP-AFDM.}
        \label{fig:eff_channel_cpafdm}
    \end{subfigure}
    \vspace{-0.5ex}
    \caption{Effective channel visualizations of the waveforms in a 3-path doubly-dispersive channel with integer normalized delays and Doppler shifts, in both 3D and 2D views, \textcolor{blue}{with $N=64$, $P=3$, $\Eric{f^{\mathrm{max}}=3}$}, illustrating the real part of the effective channel coefficients.}
    \label{fig:eff_channel_vis}
\end{figure*}
%

%




The effective channel of a waveform describes the full input-output relationship between the transmit symbols and the demodulated symbols, and provides various insights of the behavior of the waveform, especially useful to analyze the delay-Doppler orthogonality and robustness properties.

To that end, first consider the effective channels of the proposed \ac{CP-AFDM} waveforms, for both the two-sided case and the one-sided case, as given in \eqref{eq:CPAFDM-effective_channel} and \eqref{eq:osCPAFDM-effective_channel}, where the convolutional matrix form of the doubly-dispersive channel intrinsically defined in \eqref{eq:IO_matrix} is given by \vspace{-0.5ex}
\begin{equation}
\mathbf{H} \triangleq \sum_{p=1}^{P} \mathbf{H}_p = \sum_{p=1}^{P} {h_p \!\cdot\! \mathbf{\Phi}_{p} \!\cdot\! \mathbf{W}^{f_p} \!\cdot\! \mathbf{L}^{\ell_p} \in  \mathbb{C}^{N \times N}},
\label{eq:conv_chan}
\vspace{-0.5ex}
\end{equation}
which is composed of $P$ distinct matrices $\mathbf{H}_p \in \mathbb{C}^{N \times N}$, each representing a unique delay-Doppler tap of the channel with the corresponding channel coefficient $h_p$ and normalized delay-Doppler shifts $\ell_p$ and $f_p$.

The effective channels can alternatively be reformulated as
\begin{equation}
\mathbf{G}_{i_1,i_2} \triangleq \mathbf{\Lambda}_{c_2,i_2}^{\vphantom{H}} \cdot \mathbf{\Xi}_{i_1}^{\vphantom{H}} \cdot \mathbf{\Lambda}_{c_2,i_2}\herm \in \mathbb{C}^{N \times N},
\label{eq:effchan_} \vspace{-2.5ex}
\end{equation}
\begin{align}
\mathbf{\Xi}_{i_1}^{\vphantom{H}} \triangleq \mathbf{F}_N \mathbf{\Lambda}_{c_1,i_1}  \Big( \sum_{p=1}^{P} h_p \mathbf{\Phi}_{p} \mathbf{W}^{f_p}\mathbf{L}^{\ell_p}\Big) \mathbf{\Lambda}_{c_1,i_1}\herm \mathbf{F}_N\herm \in  \mathbb{C}^{N \times N}, \nonumber \\[-3ex]
\label{eq:intermediate_eff_channel_tsCPAFDM}
\end{align}

\noindent and
\begin{equation}
\mathbf{G}_{i_2} \triangleq \mathbf{\Lambda}_{c_2,i_2}^{\vphantom{H}} \cdot \mathbf{\Xi}^{\vphantom{H}} \cdot \mathbf{\Lambda}_{c_2,i_2}\herm \in \mathbb{C}^{N \times N},
\end{equation}
\begin{equation}
\mathbf{\Xi}^{\vphantom{H}} \triangleq \mathbf{F}_N \mathbf{\Lambda}_{c_1}  \Big( \sum_{p=1}^{P} h_p \mathbf{\Phi}_{p} \mathbf{W}^{f_p}\mathbf{L}^{\ell_p}\Big) \mathbf{\Lambda}_{c_1}\herm \mathbf{F}_N\herm \in  \mathbb{C}^{N \times N}.
\label{eq:intermediate_eff_channel_osCPAFDM}
\end{equation}
where the intermediate structure matrices $\mathbf{\Xi}_{i_1}$ and $\mathbf{\Xi}$ of the effective channels, represent the effective channel without the effect of the second diagonal chirp $\mathbf{\Lambda}_{c_2,i_2}$ present in the \ac{CP}-\ac{DAFT}/\ac{IDAFT}.

It is important to highlight that the transformation performed by the second diagonal chirp matrices $\mathbf{\Lambda}_{c_2,i_2}$ and $\mathbf{\Lambda}_{c_2,i_2}\herm$ onto the structure matrices $\mathbf{\Xi}_{i_1}$ and $\mathbf{\Xi}$, is only a similarity transformation, i.e., does not change the structure of the effective channel (positions of the non-zero elements), but only scales each element.

Specifically, the intermediate effective channel is scaled element-wise, as 
%
\begin{equation}
\big[\mathbf{G}_{i_1,i_2}^{\vphantom{H}}\big]_{n,n'} = e^{-j2\pi c_2 \big(\Pi_{i_2}(n)^2 - \Pi_{i_2}(n')^2\big)} \cdot \big[\mathbf{\Xi}_{i_1}^{\vphantom{H}}\big]_{n,n'},
\label{eq:effchane_elem_tscpafdm}
\end{equation}
\begin{equation}
\big[\mathbf{G}_{i_2}^{\vphantom{H}}\big]_{n,n'} = e^{-j2\pi c_2 \big(\Pi_{i_2}(n)^2 - \Pi_{i_2}(n')^2\big)} \cdot \big[\mathbf{\Xi}^{\vphantom{H}}\big]_{n,n'},
\label{eq:effchan_element_osCPAFDM}
\end{equation}
respectively for the two-sided and the one-sided \ac{CP-AFDM}.

For comparison, the reformulated form of the conventional \ac{AFDM} effective channel $\mathbf{G}^\mathrm{AFDM}$ in \eqref{eq:eff_afdm}, is given by
\begin{equation}
\mathbf{G}^\mathrm{AFDM} \triangleq \mathbf{\Lambda}_{c_2}^{\vphantom{H}} \cdot \mathbf{\Xi}^{\vphantom{H}} \cdot \mathbf{\Lambda}_{c_2}\herm \in \mathbb{C}^{N \times N},
\label{eq:effchan_AFDM}
\end{equation}
\begin{align}
\big[\mathbf{G}_{\mathrm{AFDM}}^{\vphantom{H}}\big]_{n,n'} = e^{-j2\pi c_2 \big((n)^2 - (n')^2\big)} \cdot \big[\mathbf{\Xi}^{\vphantom{H}}\big]_{n,n'}.
\label{eq:effchan_elem_AFDM}
\end{align}

The above reveals that the structure of the effective channel of the one-sided \ac{CP-AFDM} $\mathbf{G}_{i_2}$ in \eqref{eq:osCPAFDM-effective_channel} is identical to that of the conventional \ac{AFDM} $\mathbf{G}^\mathrm{AFDM}$ in \eqref{eq:effchan_AFDM} -- where both are element-wise scaled versions of $\mathbf{\Xi}$ -- while the effective channel of the two-sided \ac{CP-AFDM} $\mathbf{G}_{i_1,i_2}$ in \eqref{eq:CPAFDM-effective_channel} is quite different, being an element-wise scaled version of a different structure matrix $\mathbf{\Xi}_{i_1}$.

\newpage
This behavior is clearly illustrated in Fig.~\ref{fig:eff_channel_vis}, where the one-sided \ac{CP-AFDM} waveform has an identical channel structure to the conventional \ac{AFDM}, as seen by the same non-zero element positions of the matrix viewed in 2D, but results in vastly different channel coefficients due to the permutation in the scaling factor of \eqref{eq:effchan_element_osCPAFDM} against that of \eqref{eq:effchan_elem_AFDM}, as seen by the completely different channel coefficients viewed in 3D (only the real part visualized).

Therefore, the one-sided \ac{CP-AFDM} waveform retains the beneficial properties of the conventional \ac{AFDM} waveform, such as orthogonality and robustness against doubly-dispersive channels and highly sparse effective channel, while also providing a new degree of freedom in the design of the transmit signal through the permutation index $i_2$.

Therefore, the effective channel structure -- and consequently the delay-Doppler robustness properties of the waveform -- is only dependent on the first ($c_1$) chirp and not the second ($c_2$) chirp.
Namely, when the first chirp sequence remains unpermuted, i.e., $i_1 = 1$, and is set to the optimal value, the two-sided Fourier transforms spread each path of the channel $\mathbf{H}$ to achieve deterministic integer orthogonality in delay-Doppler domain with fractional Doppler shift-dependent sideband, as analyzed in \cite{Rou_SPM24}.

Specifically, the center of the $p$-th band-diagonal is related to the location index
\begin{align}
\label{eq:locationindex}
\mathrm{loc}_p & \triangleq \big[f_p + 2Nc_1\ell_p\big]_{\mathrm{mod}\, N} \\[1ex]
& = \Big[f_p + \ell_p\big(1 + 2(f^\mathrm{max} + \xi) \big)\Big]_{\mathrm{mod}\, N}, \nonumber 
\end{align}
where the latter reformulation in terms of the path delay and Doppler indices and \ac{AFDM} parameters is obtained under the optimal parametrization of $c_1$ as described in Section~\ref{subsec:AFDMwaveform}.

The $p$-th location index $\mathrm{loc}_p$ describes the position of the non-zero channel coefficient element in the first row (or column) of the effective matrix corresponding to the $p$-th path\footnotemark, and is a one-to-one mapping of the integer delay-Doppler indices within the unambiguous range supported by the \ac{AFDM} parameters, i.e., $\ell_p \in \{0,\ldots,\ell^\mathrm{max}\}$ and $f_p \in \{-f^\mathrm{max},\ldots,f^\mathrm{max}\}$ to an index in $\{0,\ldots,N-1\}$.
\textcolor{blue}{Note that the self-interference caused by direct-path is confined to the known diagonal at $\mathrm{loc}_p=0$ and remains separable\footnotemark.}
\addtocounter{footnote}{-1}\footnotetext{For fractional normalized Doppler shift cases, the location index describes the center (highest peak) of the band diagonals, in which the power decaying away from the central diagonal \cite{Rou_SPM24}.}
\addtocounter{footnote}{1}\footnotetext{\textcolor{blue}{However, one of the key properties of chirp-based waveform, low sampling-rate reception, is unavailable for \ac{CP-AFDM} since the de-chirped signal retains the data-bearing component occupying the full bandwidth.}}


In addition, due to the deterministic nature of the effective channel structure, identical to that of the conventional \ac{AFDM}, all algorithms developed for the conventional \ac{AFDM} using known \ac{CSI}, such as symbol detection, channel estimation, \ac{ISAC}, etc., can be directly used for the one-sided \ac{CP-AFDM} waveform.
\textcolor{blue}{The \ac{CP-AFDM} receiver accordingly consists of the matched \ac{CP}-\ac{DAFT} front-end, which differs from the \ac{AFDM} demodulator only in the ordering of a stored unit-modulus sequence, followed by classical \ac{AFDM} detector applied without modification.}

\Eric{If the permutation index is known at the receiver, the demodulation requires no additional computational complexity, only small space complexity to store the unit-modulus sequence, whereas if the permutation index is unknown, it leads to a non-trivial receiver problem of the joint detection of data symbols and permutation indices, which will be elaborated in Section~\ref{sec:application_cpi_modulation}.}


On the other hand, the two-sided \ac{CP-AFDM} waveform, i.e., $i_1 \neq 1$, results in a completely different and non-beneficial effective channel structure as illustrated in Fig.~\ref{fig:eff_channel_vis}(c), which does not have beneficial structure for channel estimation and \ac{ISAC} functionalities.
Therefore, while valid modulation and detection can be performed in theory with the two-sided \ac{CP-AFDM} waveform, it is not useful in practice especially in doubly-dispersive channels as it results in a dense effective channel, and therefore no additional benefits, compared to the one-sided case and even the conventional waveforms.

Therefore, from here on, we only consider the one-sided \ac{CP-AFDM} waveform and refer to it simply as the proposed \ac{CP-AFDM} for brevity, and provide further analysis of it as an enhanced waveform for \ac{ISAC} in doubly-dispersive channels, with a novel dimension of freedom in the design of the transmit signal through the permutation index $i_2$.

{\subsection{BER Performance}
\label{subsec:performance_analysis_BER}

\begin{table}[b]
    \vspace{-2ex}
\centering
\renewcommand*{\arraystretch}{0.85}
\captionsetup{justification=centering}
\caption{\textcolor{blue}{\Eric{Numerical simulation parameters}}}
\label{Parameters}
\color{blue}
\begin{tabular}{c | c | c}
\Xhline{1\arrayrulewidth}
~~~Description~~~ & ~~~Parameter~~~ & ~~~Value~~~\\
\hline \hline
Carrier frequency &  $-$ & $50~\mathrm{GHz}$ \\
\hline
Bandwidth &  $B$ & $150~\mathrm{kHz}$ \\
\hline
Number of subcarriers & $N$ & $64$ \\
\hline
Number of paths & $P$ & $3$ \\ 
\hline
Sampling frequency & $f_\mathrm{s}$ & $300~\mathrm{kHz}$ \\ 
\hline
Symbol type & $-$ & $\mathrm{QAM}$ \\ 
\hline
\Eric{Modulation order} & $M$ & $4$ \\ 
\hline
Normalized maximum delay & $\ell^{\max }$ & $6$ \\ 
\hline
Normalized maximum Doppler shift & $f^{\mathrm{max}}$ & $3$ \\
\hline
\Eric{Guard width} & \Eric{$\xi$} & \Eric{$1$} \\
\hline
\Eric{Second chirp frequency} & \Eric{$c_2$} & \Eric{$1/(2\pi N)$} \\
\hline
\Eric{Number of time-domain symbols} & \Eric{$-$} & \Eric{$1$} \\
\hline
\end{tabular}
\vspace{2ex}
\end{table} 

As previously analyzed, the proposed \ac{CP-AFDM} preserves the same effective channel structure and behavior as conventional \ac{AFDM}, and its communication performance in doubly-dispersive environments is therefore expected to be equivalent.
\Eric{In this section, we provide simulation results to validate this expectation, and also compare the performance of \ac{CP-AFDM} with other waveforms. 
Unless otherwise specified, the simulation parameters are organized in Table~\ref{Parameters}.}

Fig.~\ref{fig:BER} presents the simulated \ac{BER} performance of \ac{CP-AFDM} in comparison with \ac{AFDM}, \ac{OTFS}, and \ac{OFDM} over a doubly-dispersive channel with $P = 3$ paths and $N = 64$ subcarriers, utilizing conventional \ac{MMSE} detection  \textcolor{blue}{under perfect \ac{CSI}}\footnote{\Eric{Note that the impact of imperfect \ac{CSI} would be the same for both of \ac{AFDM} and \ac{CP-AFDM} since the channel matrices are unitarily identical. The effects and joint channel estimation techniques have been well-analyzed in the literature \cite{ranasinghe2024joint}, and is omitted in our simulation results for brevity.}} with matched demodulation.
As shown, \ac{CP-AFDM} achieves the same full diversity and \ac{BER} performance as \ac{AFDM}, due to its inherent delay-Doppler orthogonality and robustness inherited from the underlying structure.

\textcolor{blue}{Note that the identical \ac{BER} performance between \ac{AFDM} and \ac{CP-AFDM} is an intended design since the permutation acts as a unitary rotation of the detection problem.}

This performance closely matches that of \ac{OTFS} at low-to-moderate \ac{SNR} levels -- however, unlike \ac{OTFS}, which is known to suffer diversity degradation at high \ac{SNR} without phase rotations or precoding \cite{surabhi2019diversity}, both \ac{AFDM} and \ac{CP-AFDM} maintain robust performance across the entire \ac{SNR} range. 
In contrast, \ac{OFDM} exhibits a significant performance loss in the presence of Doppler spread, primarily due to inter-carrier interference (\ac{ICI}), highlighting the superior robustness of next-generation waveforms.

\begin{figure}[t]
\vspace{2ex}
\centering
\includegraphics[width=\columnwidth]{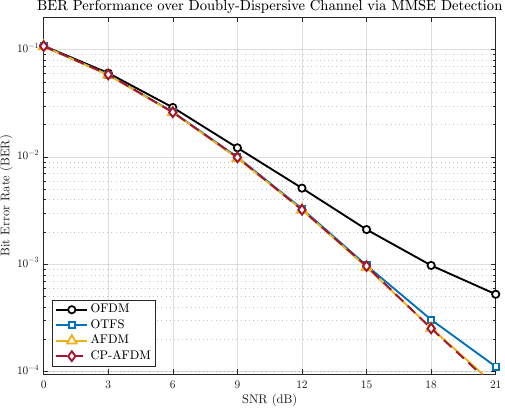}
\caption{BER performance of the proposed CP-AFDM and various waveforms over a doubly-dispersive channel with $P = 3$ paths and $N = 64$ subcarriers, under MMSE detection with matched demodulator and perfect CSI.
\textcolor{blue}{Here, the Doppler spread and the corresponding velocity range are $[-\nu_{\max},\nu_{\max }]=[-14.06,14.06]~\mathrm{kHz}$ and $[-84.3,84.3]~\mathrm{m/s}$, respectively.}}
\label{fig:BER}
\end{figure}
}

\subsection{CSI-based Radar Parameter Estimation Performance}
\label{subsec:performance_analysis_ISAC}

Another important implication of the deterministic effective channel structure and delay-Doppler orthogonality, also discussed in Section~\ref{subsec:effective_channel_analysis} and \eqref{eq:locationindex}, is the applicability of \ac{CSI}-based radar parameter estimation, or conversely, radar-assisted \ac{CSI} acquisition. 
This concept, in the context of delay-Doppler representative waveforms such as \ac{OTFS} and \ac{AFDM}, has gained attention and been thoroughly explored in the literature via a plethora of methods \cite{muppaneni2023channel,ranasinghe2024fast,ranasinghe2024joint}.

These methods exploit the fact that the effective channel matrices of \ac{AFDM}, \ac{OTFS}, and also the proposed one-sided \ac{CP-AFDM} are fully characterized by the set of normalized delay indices $\ell_p$, normalized digital Doppler shifts $f_p$, and complex channel gains $h_p$. 
Accordingly, the full effective channel estimation problem reduces to the joint estimation of the $3P$ parameters: $\{\ell_1, \ldots, \ell_P\}$, $\{f_1, \ldots, f_P\}$, and $\{h_1, \ldots, h_P\}$.

To this end, the radar parameter estimation performance of \ac{CP-AFDM} is evaluated in comparison to \ac{AFDM} and \ac{OTFS}, employing the state-of-the-art method of~\cite{ranasinghe2024joint}, assuming a pilot-based transmission and leveraging a \ac{PDA} estimator for the radar parameters.

\Eric{Note that the \ac{OFDM} is not included in the radar parameter estimation performance evaluation of Fig.~\ref{fig:RadarParam}, since the adopted \ac{PDA} method is a single-symbol-based joint channel-radar estimation approach not applicable to \ac{OFDM}, which requires multiple symbols to estimate the Doppler shift.}

Consistent with the \ac{BER} analysis, the \ac{CP-AFDM} retains the full delay-Doppler representation of the effective channel and \Eric{only rescales its elements by unit-modulus factors, as shown in \eqref{eq:effchan_element_osCPAFDM}}. 
As a result, the delay-Doppler statistics remain unchanged from those of \ac{AFDM}, leading to virtually identical \ac{RMSE} performance in radar parameter estimation, in both range and velocity, as shown in Fig.~\ref{fig:RadarParam}.
\textcolor{blue}{This result highlights the sensing advantage of \ac{CP-AFDM}, which can be organized as follows:
First, since the \ac{CP-AFDM} inherits the full delay-Doppler characteristics of the \ac{AFDM} effective channel, both the single-frame and CSI-based \ac{JCDE} are available, which is inapplicable for \ac{OFDM} since \Eric{paths with equal Doppler shifts overlap} in its effective channel~\cite{ranasinghe2024joint}.
\Eric{Secondly, all chirp permutations can be applied freely without any cost \Eric{to the CSI-based} sensing performance.
This implies that the additional degrees-of-freedom by the chirp permutations can be exploited for other multifunctional purposes while retaining this sensing performance.
As one example, \ac{CP-AFDM} can be integrated into the index modulation and achieves enhanced spectral efficiency without any additional transmit energy or sparsity in subcarriers, which is a bottleneck in conventional index modulation, where its details are summarized in Section~\ref{sec:application_cpi_modulation}.}}

These findings also naturally motivate the suitability of \ac{CP-AFDM} for advanced \ac{ISAC} applications, where the joint communication and radar tasks are performed under limited pilot overhead and with partially unknown data streams, as indeed, under the same framework \cite{ranasinghe2024joint} and others \cite{luo2025novel,luo2025target,zhu2024afdm}, \ac{AFDM} facilitates robust \ac{JCDE} through high-resolution estimation of delay and Doppler parameters, and hence also expected of the proposed waveform.

\begin{figure}[t]
\centering
\includegraphics[width=1\columnwidth]{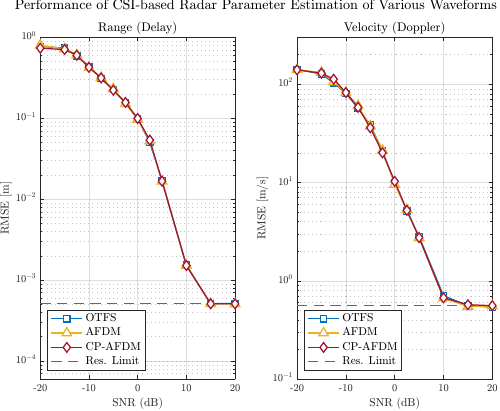}
\caption{Radar parameter estimation performance of the proposed \ac{CP-AFDM} waveform compared to \ac{AFDM} and \ac{OTFS}, in terms of \ac{RMSE} for target range and velocity, via the \ac{SotA} \ac{PDA}-based approach~\cite{ranasinghe2024joint} with $N = 64$.}
\label{fig:RadarParam}
\vspace{-3ex}
\end{figure}

\subsection{\Acf{PAPR} Analysis}
\label{subsec:papr_analysis}

The \ac{PAPR} is a critical metric for waveform evaluation, particularly in multicarrier and chirp-based systems where signal peaks can lead to nonlinear distortion and power inefficiency in practical \ac{RF} chains. 
This subsection presents a detailed mathematical formulation and comparative study of the PAPR characteristics of the proposed \ac{CP-AFDM} waveform (one-sided case), benchmarked against \ac{AFDM} and \ac{OFDM}.

The \ac{PAPR} of the \ac{CP-AFDM} signal with permutation parameters $i_1,i_2$, is given by
\begin{align}
\mathrm{PAPR} &\triangleq \frac{\displaystyle\max_{n} \big|s[n]\big|^2}{\textstyle \frac{1}{N} \sum_{n=0}^{N-1} \big|s[n]\big|^2}\label{eq:cpafdm_papr_def} \\
& = \frac{\displaystyle\max_{n} \big| \textstyle\sum_{m=0}^{N-1} x_m \cdot \kappa_n(m) \big|^2}{\textstyle \frac{1}{N} \sum_{n=0}^{N-1} \big|\sum_{m=0}^{N-1} x_m \cdot \kappa_n(m)\big|^2}. \nonumber
\end{align}

In line with conventional \ac{PAPR} analysis, we assume that the transmit symbols $x_m$ are \ac{i.i.d.} complex random variables with zero mean and variance $\sigma_x^2$, i.e., $x_m \sim \mathcal{CN}(0, \sigma_x^2), \,\forall m$.

\Eric{Then, since $s[n]$ is a linear combination of the \ac{i.i.d.} Gaussian symbols $x_m$ with deterministic coefficients $\kappa_n(m)$, it is itself a zero-mean complex Gaussian random variable, as}
\begin{equation}
s[n] \sim \mathcal{CN}\left(0, \sigma_x^2 \textstyle\sum_{m=0}^{N-1} |\kappa_n(m)|^2\right) = \mathcal{CN}\left(0, \sigma_x^2\right),
\end{equation}
where $|\kappa_n(m)|^2 = \tfrac{1}{N}$ is trivial from \eqref{eq:kernal_CPAFDM}.

Following the above, the instantaneous power samples of the modulated signal follow an exponential distribution,
%
which then, by the Law of Large Numbers, the \ac{PAPR} of the \ac{CP-AFDM} signal converges in distribution to the maximum of $N$ \ac{i.i.d.} exponential random variables,
\begin{equation}
\mathrm{PAPR} \xrightarrow{N \to \infty} {\displaystyle\max_{n} \, \tfrac{1}{\sigma_x^2} \tilde{Z}_n} = \max_n {Z_n},
\end{equation}
where $\tilde{Z}_n \sim \mathrm{Exp}(\sigma_x^2)$ and $Z_n \sim \mathrm{Exp}(1)$.

The \ac{CDF} of the maximum is then given by
\begin{equation}
F_{\mathrm{PAPR}}(\gamma) = \mathbb{P}\left[ \max_n Z_n \leq \gamma \right] = \left(1 - e^{-\gamma} \right)^N,
\end{equation}
and the \ac{CCDF}, often used in \ac{PAPR} analysis, is given by
\begin{equation}
\mathbb{P}[\mathrm{PAPR} > \gamma] = 1 - F_{\mathrm{PAPR}}(\gamma) = 1 - \left(1 - e^{-\gamma} \right)^N.
\label{eq:papr_ccdf}
\end{equation}

The final formulation of \eqref{eq:papr_ccdf} is identical to the asymptotical \ac{PAPR} of general waveforms with Gaussian-approximated symbols, whose time-domain representation is obtained through a unitary transformation, i.e., to that of the \ac{OFDM} and conventional \ac{AFDM} waveforms.

The above analysis is verified in Fig.~\ref{fig:PAPR_CCDF}, where the \ac{CCDF} of the \ac{PAPR} of the proposed \ac{CP-AFDM} waveform is evaluated numerically and compared against the \ac{OFDM}, \ac{OTFS}, and the conventional \ac{AFDM} waveforms under varying system sizes and constellations.
The \ac{PAPR} behavior is shown to be identical to that of the conventional \ac{OFDM} and \ac{AFDM} waveforms, illustrated with overlapping \ac{CCDF} curves, while the \ac{OTFS} with varying grid sizes can improve the \ac{PAPR} \cite{surabhi2019peak}.

\Eric{The \ac{PAPR} equivalence with \ac{AFDM} and \ac{OFDM} holds for every pair of permutation indices $i_1$ and $i_2$, since the derivation above relies only on the unitarity of the \ac{CP}-\ac{DAFT}, which Theorem~\ref{thm:cpdaft_unitary} guarantees for all $i_1$ and $i_2$.
The first permutation $i_1$ leaves the \ac{PAPR} unchanged even for each realization of the transmit symbols, as the phase-only matrix $\mathbf{\Lambda}_{c_1,i_1}\herm$ in \eqref{eq:txsig_cpafdm} preserves every sample magnitude $|s[n]|$.}

\begin{figure}[t]
\begin{subfigure}[c]{1\columnwidth}
\centering
\includegraphics[width=0.97\columnwidth]{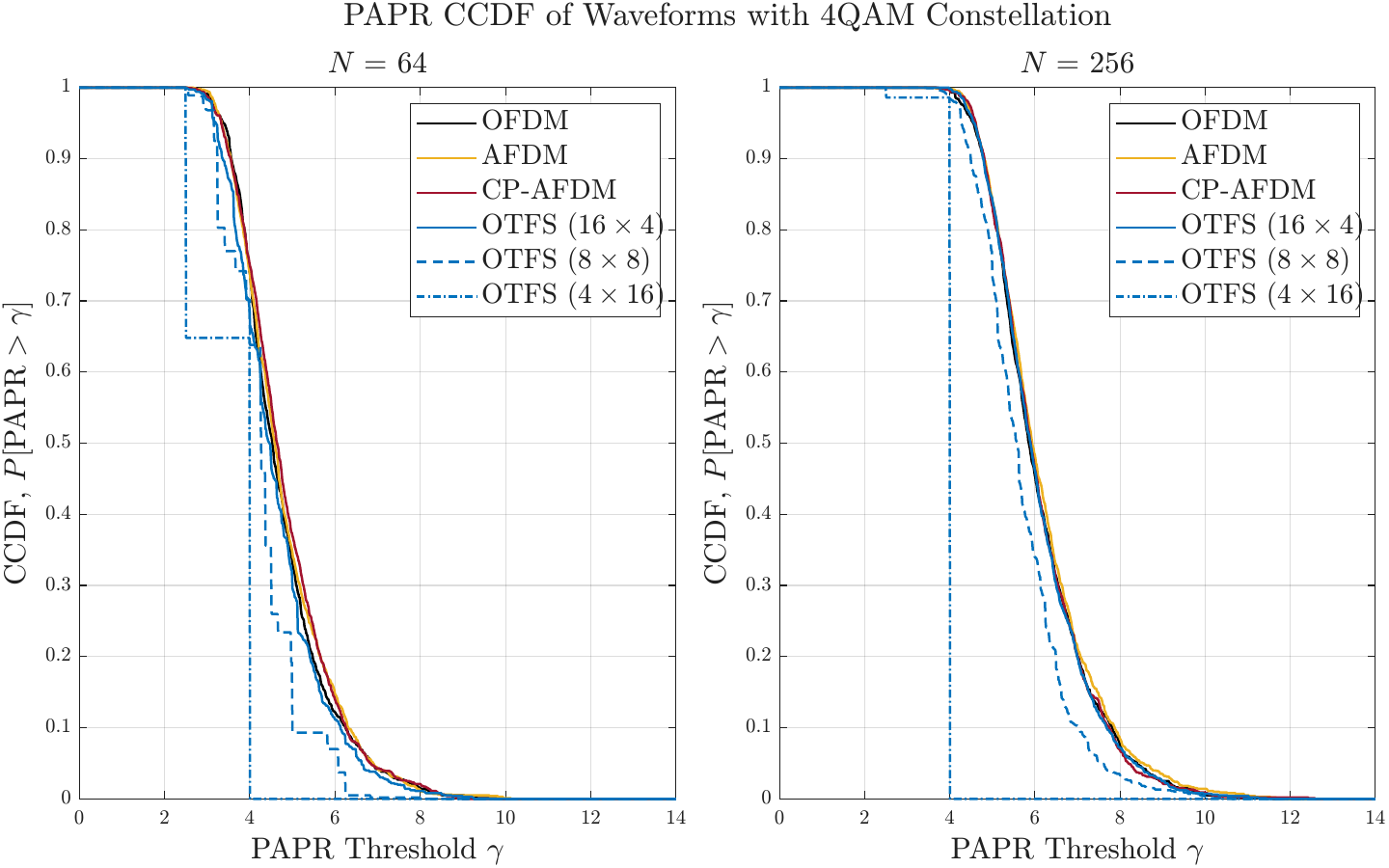}
\vspace{-1ex}
\subcaption{4-QAM symbol constellation.}
\label{fig:PAPR_CCDF_4QAM}
\end{subfigure}
\vspace{1.5ex}

\begin{subfigure}[c]{1\columnwidth}
    \centering
\includegraphics[width=0.97\columnwidth]{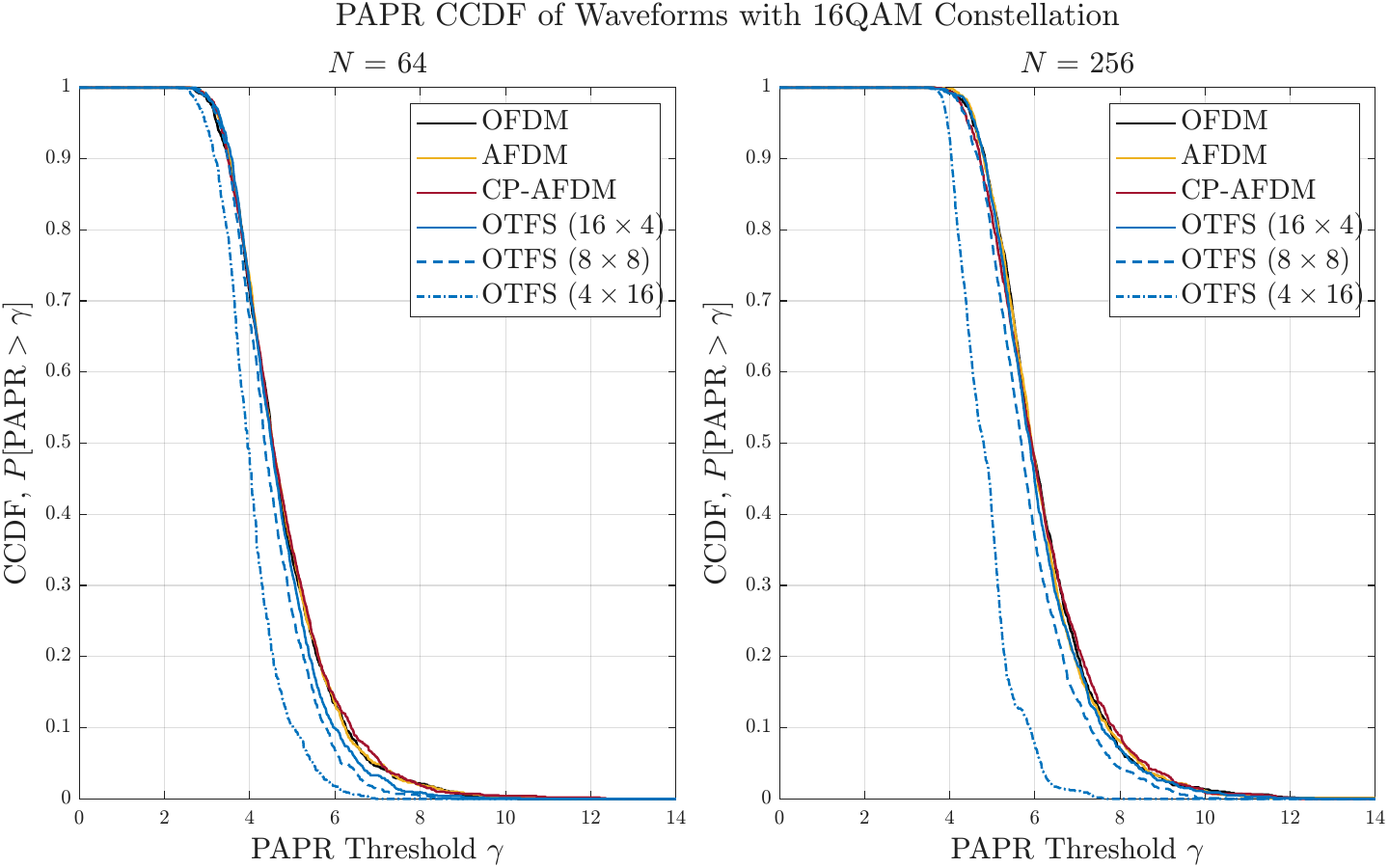}
\vspace{-1ex}
\subcaption{16-QAM symbol constellation.}
\label{fig:PAPR_CCDF_16QAM}
\end{subfigure}
%
%
\caption{CCDF of the waveform \ac{PAPR} of the OFDM, OTFS (for different grid sizes), conventional AFDM, and the proposed \ac{CP-AFDM} under varying system size $N$ and symbol constellation. \textcolor{blue}{The parameters are set as $f^{\mathrm{max}}=2$, $c_{1}=\frac{7}{2N}$, and $c_{2}=\frac{1}{2\pi N}$.}}
\label{fig:PAPR_CCDF}
\vspace{-3ex}
\end{figure}

\vspace{-2ex}
\subsection{Ambiguity Function (AF) Analysis}

The \ac{AF} is a crucial analytical tool for evaluating the delay-Doppler characteristics of waveforms, especially to be utilized for radar sensing and \ac{ISAC} functionalities.
Mathematically, for a given transmitted signal $s(t)$ in the continuous time domain, the \ac{AF} is defined as
\begin{equation}
A(\tau, \nu) = \left| \int_{-\infty}^{\infty} s(t)s^*(t-\tau)e^{-j2\pi\nu t}\,dt \right|,
\end{equation}
where $\tau$ represents the time delay, and $\nu$ denotes the Doppler frequency shift, which is typically expressed in terms of the normalized delay \Eric{$\frac{\tau}{NT_s}$} and normalized Doppler frequency \Eric{$\frac{N\nu}{f_s}$} in digital implementations with sampled signals, where $T_s \triangleq \frac{1}{f_s}$ is the sampling period with sampling frequency $f_s$.

The zero-Doppler $A(\tau,0)$ and zero-delay $A(0,\nu)$ cuts of the \ac{CP-AFDM} \ac{AF} are compared against various \ac{SotA} waveforms including \ac{OFDM}, \ac{OTFS}, and the conventional \ac{AFDM}, which highlight the Doppler and delay resolution capabilities of the waveform.

\textcolor{blue}{The \ac{AF} is analyzed for the cases of all-ones transmit symbols in Fig.~\ref{fig:AF_ones}, and of random transmit symbols, averaged over 1,000 independent realizations of normalized 4-\ac{QAM} symbols in Fig.~\ref{fig:AF_random}.}
\textcolor{blue}{For \ac{AF} shaping with \ac{CP-AFDM}, $i_{2}$ is selected offline.}
\Eric{The multi-symbol \ac{AF}, whose longer observation narrows the Doppler mainlobe of all considered waveforms, is outside the scope of this work and has been studied in~\cite{niu2025afdm}.}

In addition, several metrics of the \ac{AF} are highlighted to quantify the delay-Doppler resolution capabilities of the proposed \ac{CP-AFDM} waveform.
%
Defining the set of delay-Doppler indices within the mainlobe as $\mathcal{W} \triangleq \{\tau, \nu : |\tau| \leq \tau_{3\mathrm{dB}}, |\nu| \leq \nu_{3\mathrm{dB}} \}$\Eric{, with $\tau_{3\mathrm{dB}}$ and $\nu_{3\mathrm{dB}}$ the $-3$~dB half-widths of the mainlobe}, the \ac{PSLR} and \ac{ISLR} of the ambiguity function are also evaluated, which are respectively defined as \vspace{1ex}
\begin{equation}
\mathrm{PSLR}_{\mathrm{dB}} = 20\log_{10}\!\left(\Eric{\frac{\max\limits_{\tau,\nu \notin \mathcal{W}} \big|A(\tau,\nu)\big|}{\big|A(0,0)\big|}}\right), \vspace{1ex}
\end{equation}
\begin{equation}
\,\mathrm{ISLR}_{\mathrm{dB}} \, = 10\log_{10}\!\left(\!\frac{~~\int\!\!\!\!\!\!\!\int\limits_{\!\!\tau,\nu \notin \mathcal{W}} \!\!\!|A(\tau, \nu)|^2 }{~~\int\!\!\!\!\!\!\!\int\limits_{\!\!\tau,\nu \in \mathcal{W}} \!\!\!|A(\tau, \nu)|^2}\,\right).\vspace{1ex}
\end{equation}

The \ac{PSLR} is defined as the ratio of the \Eric{highest sidelobe amplitude to the mainlobe peak amplitude} in the ambiguity function, quantifying the worst-case ambiguity between the mainlobe and potential false detections. 
A lower \ac{PSLR} is desirable to ensure robust isolation of the mainlobe, thereby minimizing the likelihood of mistaking sidelobes for true targets or signal components. 
On the other hand, the \ac{ISLR} is defined as the ratio of the total energy in all sidelobes to the energy in the mainlobe.
A low \ac{ISLR} indicates that the energy of the waveform is predominantly concentrated in the mainlobe, reducing interference from distributed ambiguities. 
Together, the two metrics provide complementary insights to the waveform characteristics in the unambiguous resolution of delay and Doppler characteristics.



As illustrated in Fig.~\ref{fig:AF_ones} and Fig.~\ref{fig:hist_PSLRISLR}, the proposed \ac{CP-AFDM} waveform exhibits a compelling ambiguity function profile, achieving a favorable trade-off between resolution and sidelobe performance compared to baseline waveforms including \ac{AFDM}, \ac{OTFS}, and \ac{OFDM}.

In particular, \ac{CP-AFDM} attains the best Doppler resolution among all considered waveforms, while matching the delay resolution of \ac{AFDM} and \ac{OFDM}.
\textcolor{blue}{This superior Doppler resolution of \ac{CP-AFDM} results from the \Eric{spread of the signal energy over the symbol duration} induced by the chirp-permuted sequence.
Since the zero-delay cut of ambiguity function is given by $A(0, \nu)=\left|\int_{-\infty}^{\infty}| s(t)|^2 e^{-j 2 \pi \nu t} d t \right|$, the Doppler resolution \Eric{improves as $|s(t)|^2$ spreads over the symbol duration}, due to Fourier duality.
Likewise, \ac{OFDM} shows the worst Doppler resolution under all-ones symbol since it becomes like a single-time domain pulse with minimal energy spread, which in turn results in zero \Eric{delay-domain} side-lobes as shown in Fig.~\Eric{\ref{fig:AF_ones_zeroDoppler}}.}
%

\begin{figure}[h]
\vspace{1ex}
\begin{subfigure}[c]{\columnwidth}
\includegraphics[width=0.97\columnwidth]{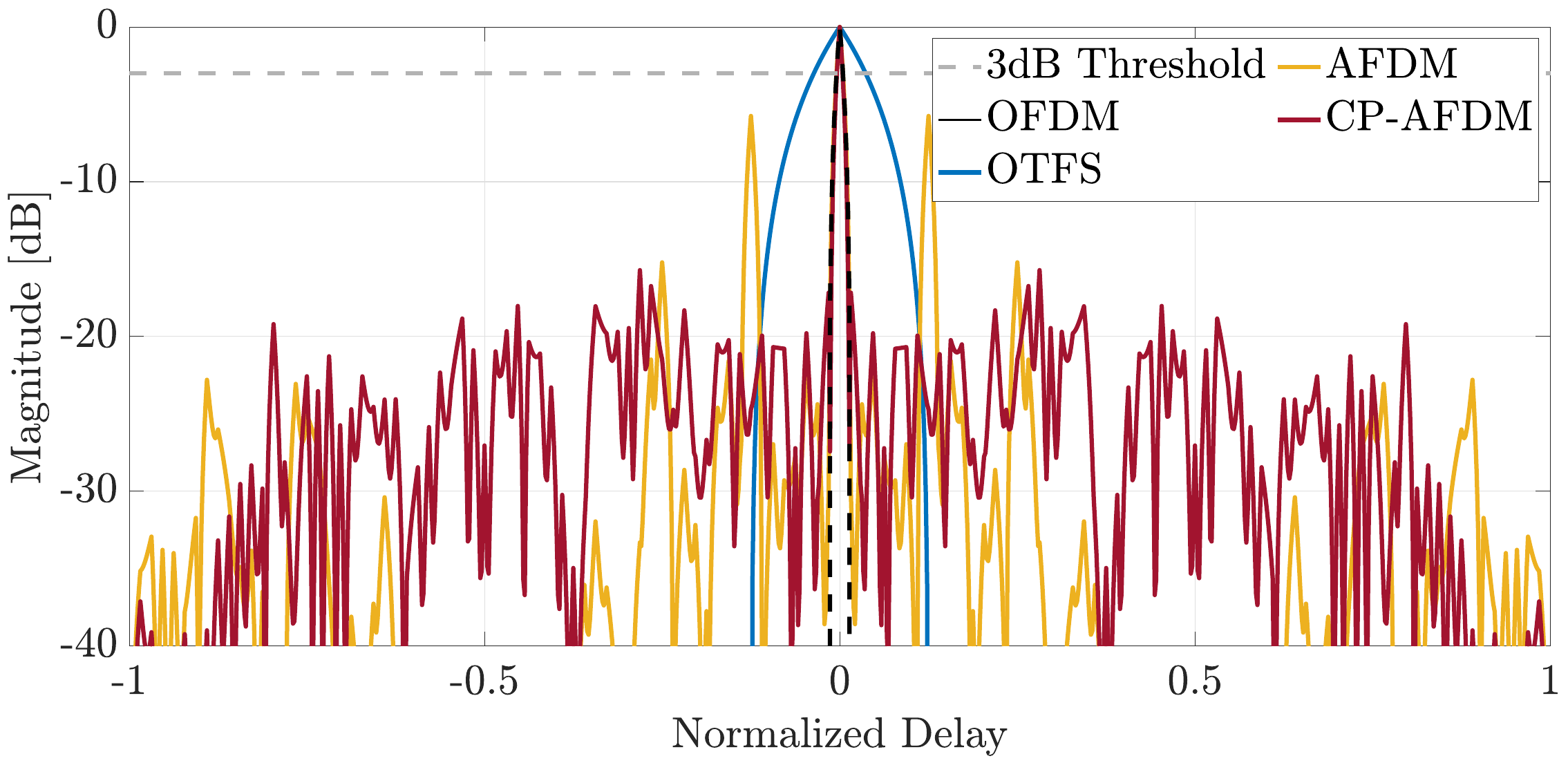}
\subcaption{Zero-Doppler cut (Delay resolution).}
\label{fig:AF_ones_zeroDoppler}
\end{subfigure}
\vspace{1ex}

\begin{subfigure}[c]{1\columnwidth}
\includegraphics[width=0.97\columnwidth]{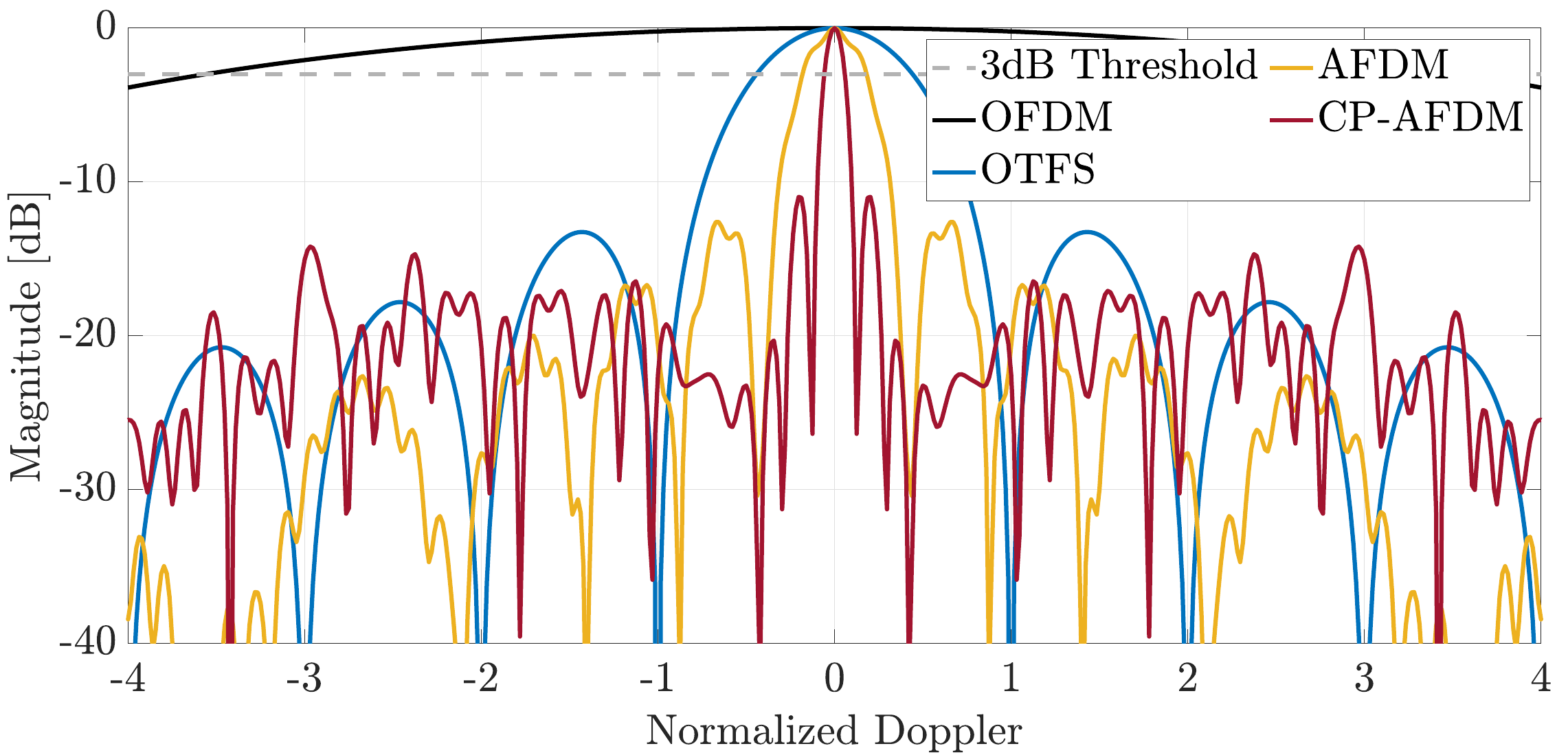}
\subcaption{Zero-Delay cut (Doppler resolution).}
\label{fig:AF_ones_zeroDelay}
\end{subfigure}
\caption{Instantaneous \acp{AF} of various waveforms under all-ones transmit symbols, with $N = 64$, \textcolor{blue}{$c_{1}=\frac{(2(f^{\mathrm{max}}+\xi)+1)}{2N}$, and $c_{2}=\frac{1}{2 \pi N}$.}}
\label{fig:AF_ones}
\vspace{-1.5ex}
\end{figure}

\begin{figure}[h]

\begin{subfigure}[c]{1\columnwidth}
\NewFig{\includegraphics[width=\columnwidth]{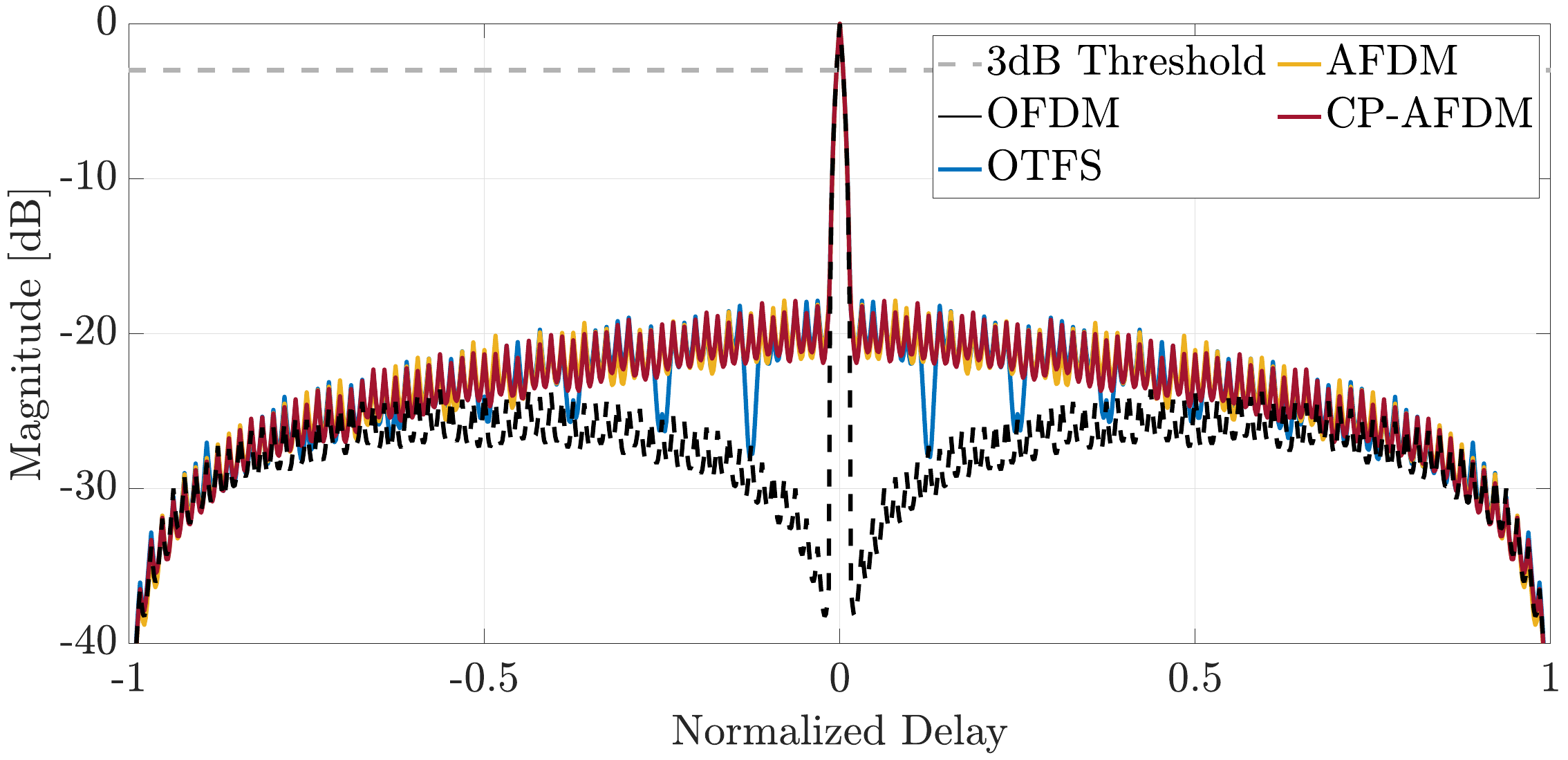}}
\subcaption{Zero-Doppler cut (Delay resolution).}
\label{fig:AF_random_zeroDoppler}
\end{subfigure}
\begin{subfigure}[c]{1\columnwidth}
\NewFig{\includegraphics[width=1\columnwidth]{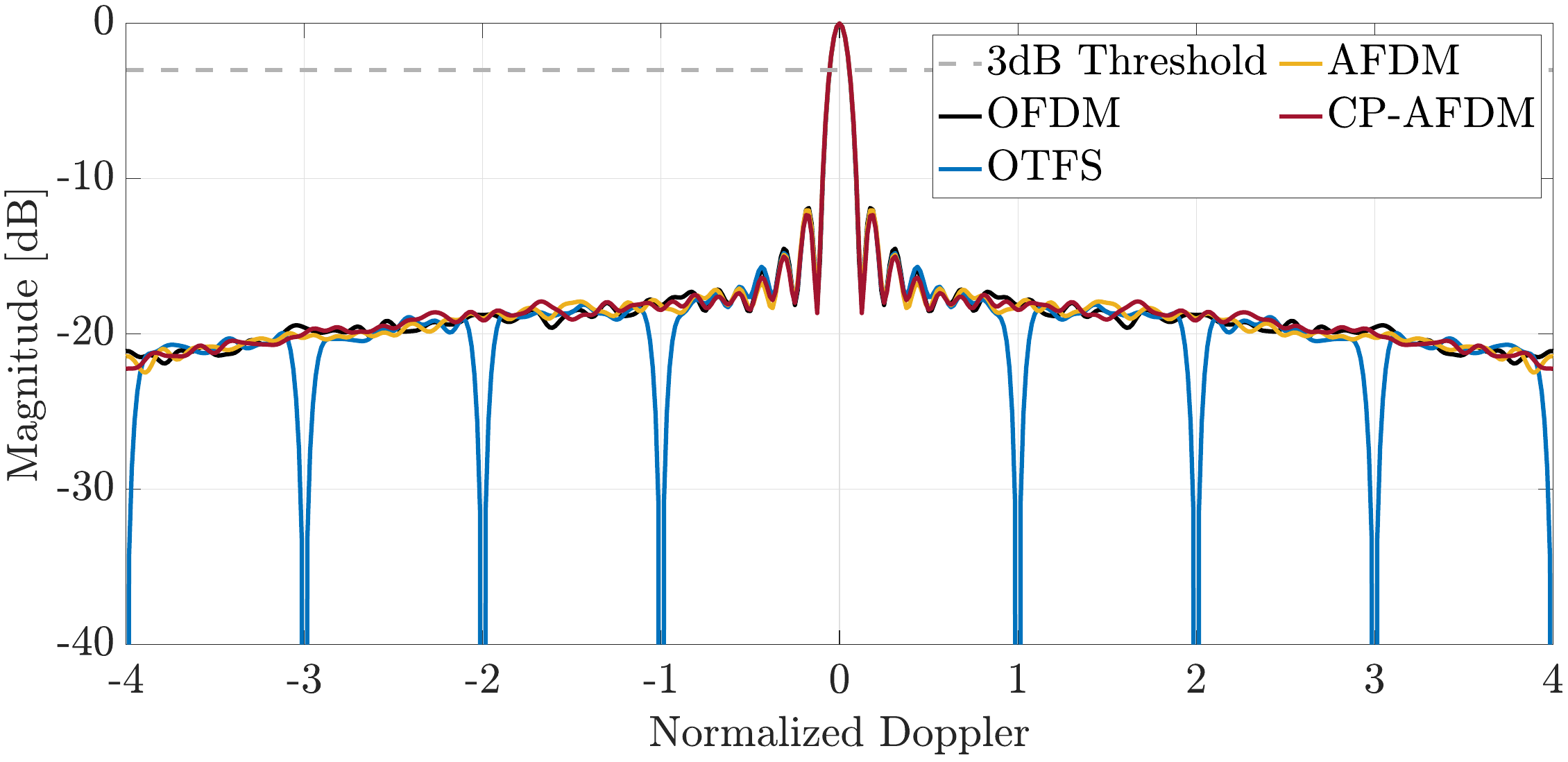}}
\subcaption{Zero-Delay cut (Doppler resolution).}
\label{fig:AF_random_zeroDelay}
\end{subfigure}
\caption{\textcolor{blue}{AFs of various waveforms under random QAM transmit symbols, with $N = 64$, $c_{1}=\frac{(2(f^{\mathrm{max}}+\xi)+1)}{2N}$, and $c_{2}=\frac{1}{2 \pi N}$.}}
\label{fig:AF_random}
\vspace{-1ex}
\end{figure}

\textcolor{blue}{Meanwhile, the delay and Doppler resolution characteristics become similar when using random symbols, as shown in Fig.~\ref{fig:AF_random}.
In case of Doppler domain, this is because the power envelopes of all waveforms are equalized under random symbols, so the Doppler resolution is primarily affected by the observation time duration.}

\textcolor{blue}{
Furthermore, \ac{OFDM} exhibits \Eric{lower delay-domain} sidelobes, which is also consistent with~\cite{liu2025cp}, whose results indicate that random transmit symbols can enhance the \ac{AF} characteristics of other \ac{SotA} waveforms, \Eric{under the requirement that the sensing receiver knows the transmitted symbols}~\cite{rou2026normalized}.} 

Regarding sidelobe behavior, \ac{CP-AFDM} demonstrates significant improvement in \ac{PSLR} across both delay and Doppler dimensions relative to \ac{AFDM} and \ac{OTFS}, highlighting enhanced suppression of sidelobes adjacent to the mainlobe.
This property is particularly desirable in radar-centric applications, as it contributes to reduced false alarm rates and improved target separation in detection tasks, including joint radar-communication systems and multipath-resolved receivers.

\begin{figure*}[t]
\begin{subfigure}[c]{0.5\columnwidth}

\includegraphics[width=\columnwidth]{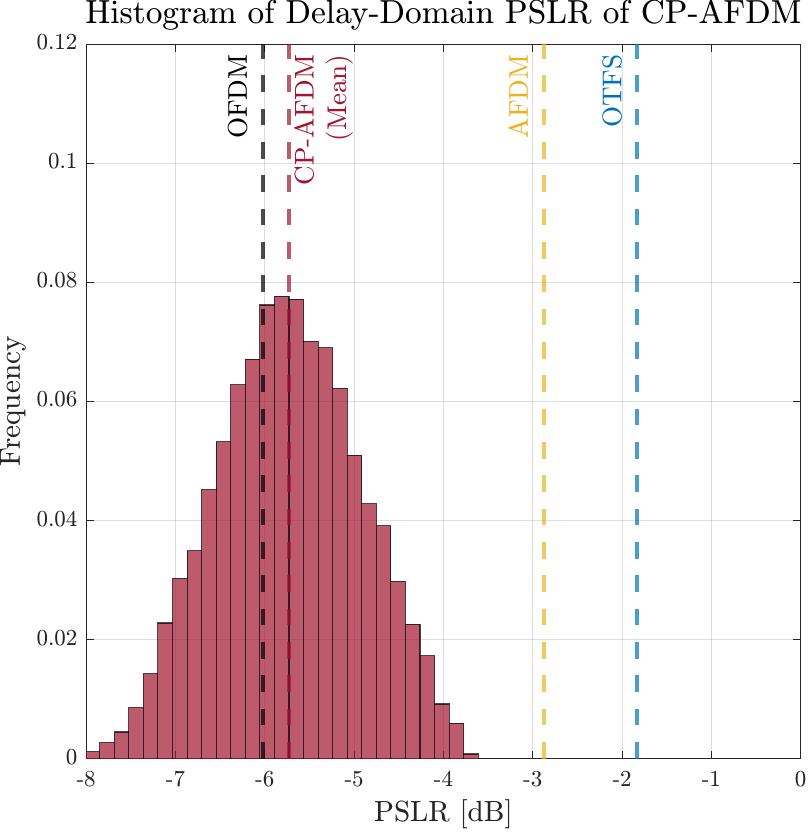}
\subcaption{Delay-domain PSLR.}
\label{fig:hist_delPSLR}
\end{subfigure}
\begin{subfigure}[c]{0.5\columnwidth}
\includegraphics[width=1\columnwidth]{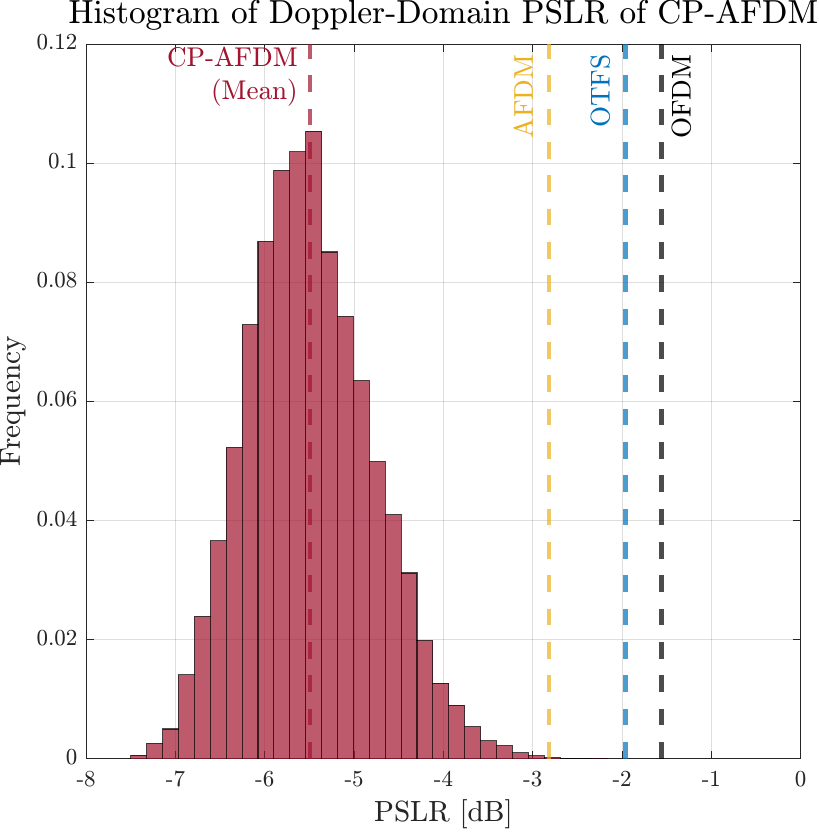}
\subcaption{Doppler-domain PSLR.}
\label{fig:hist_doppPSLR}
\end{subfigure}
\begin{subfigure}[c]{0.5\columnwidth}

\includegraphics[width=\columnwidth]{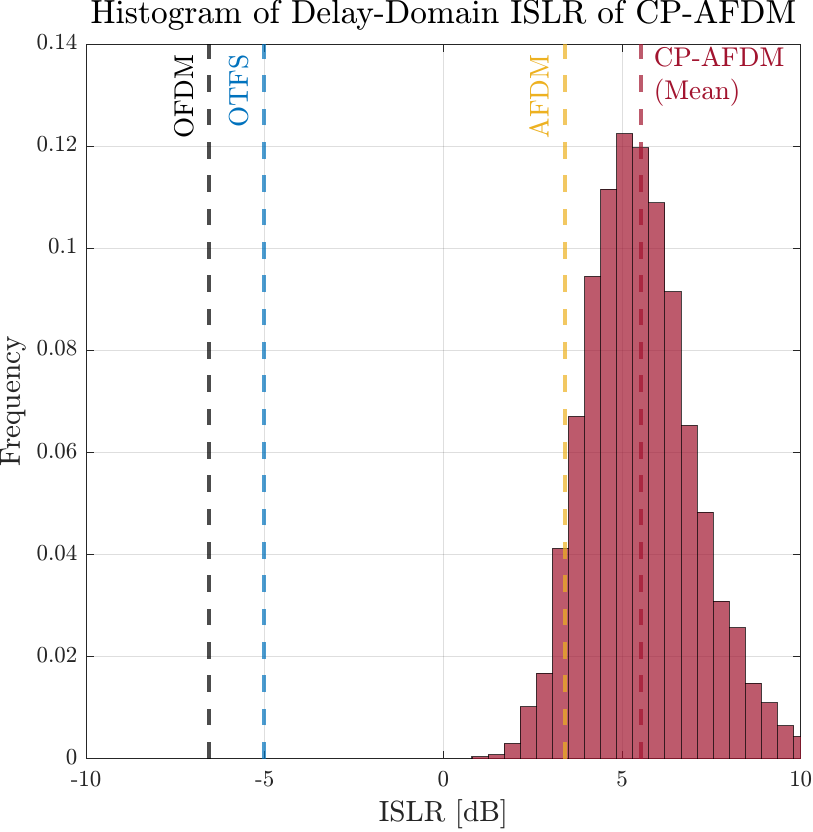}
\subcaption{Delay-domain ISLR.}
\label{fig:hist_delISLR}
\end{subfigure}
\begin{subfigure}[c]{0.5\columnwidth}
\includegraphics[width=1\columnwidth]{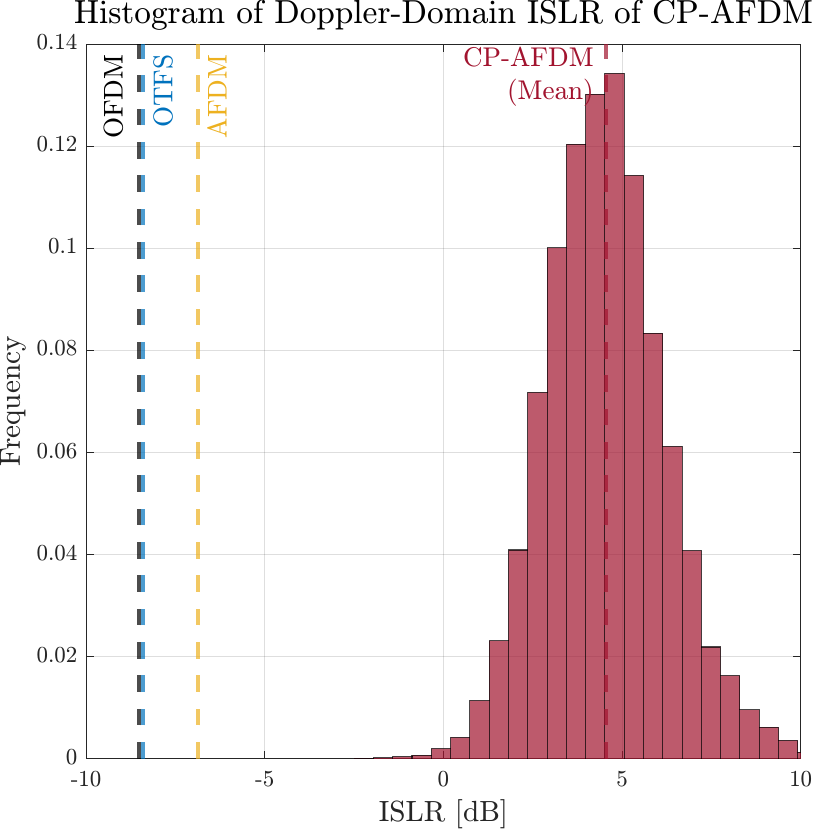}
\subcaption{Doppler-domain ISLR.}
\label{fig:hist_doppISLR}
\end{subfigure}

\caption{Histogram of PSLR and ISLR of CP-AFDM.}
\label{fig:hist_PSLRISLR}
\vspace{-2ex}
\end{figure*}

However, this gain in \ac{PSLR} is accompanied by a degradation in \ac{ISLR}, reflecting increased energy spread across distant sidelobes. 
\textcolor{blue}{This is because the chirp permutation renders the modulated samples approximately independent complex Gaussian variables, and the instantaneous power envelope is exponentially distributed about its mean.
The structured sidelobe peaks adjacent to the mainlobe are thereby replaced by this floor, which improves the PSLR.}
While \ac{ISLR} remains an important metric, its impact on practical systems is less critical than that of \ac{PSLR}, as it typically affects only distant sidelobe energy that can be mitigated through receiver-side filtering, sidelobe-aware detection, or windowing \cite{tan2025two,azouz2021design}. 
These approaches are more tractable than compensating for poor resolution or elevated mainlobe ambiguity.

Furthermore, as shown in the histograms of Fig.~\ref{fig:hist_PSLRISLR}, the \ac{PSLR} and \ac{ISLR} performance of \ac{CP-AFDM} depends on the specific permutation applied. 
The mean values correspond to empirical averages across different permutations, yet the comparative advantage remains evident, as the \Eric{\ac{PSLR}} distributions of \ac{CP-AFDM} lie apart from those of \Eric{\ac{AFDM} and \ac{OTFS}}, which are permutation-invariant.
It is also noteworthy that the mainlobe width, and hence the delay and Doppler resolution, remains unaffected by the choice of permutation in \ac{CP-AFDM}, ensuring consistent resolution regardless of sidelobe variation.

In all, the proposed \ac{CP-AFDM} waveform offers a strategically beneficial ambiguity profile, trading a modest increase in overall sidelobe energy for substantial gains in \Eric{Doppler} resolution and main sidelobe suppression. 
This makes it a strong candidate for future wireless systems requiring both spectral efficiency and \Eric{sensing} robustness.

Finally, the permutation parameter could be optimized to tailor the ambiguity function to application-specific requirements.
\textcolor{blue}{For example, \Eric{the \ac{PSLR} minimization problem under an \ac{ISLR} constraint} can be formulated as
\begin{equation} 
\hat{i}_{2}=\underset{i_2 \in\{1, \ldots, N!\}}{\Eric{\operatorname{argmin}}} \operatorname{PSLR}_{\mathrm{dB}}\left(i_2\right)  \text { s.t. }  \operatorname{ISLR}_{\mathrm{dB}}\left(i_2\right) \leq \epsilon
\label{opt_PSLR_ISLR}
\end{equation}}

{\color{blue}\noindent where $\operatorname{PSLR}_{\mathrm{dB}}\left(i_2\right)$ and $\operatorname{ISLR}_{\mathrm{dB}}\left(i_2\right)$ are PSLR and ISLR evaluated on the ambiguity function of the CP-AFDM signal $\mathbf{s}_{i_{2}}$ and $\epsilon$ is a tolerable sidelobe energy level.
Determining the optimal permutation by solving \eqref{opt_PSLR_ISLR} encompasses several important research branches, such as designing low-complexity solvers, analysis of convergence, and the characterization of the PSLR-ISLR trade-off, which remains as future work.
}

\subsection{\textcolor{blue}{Computational Complexity Analysis}}
\label{sec:complexity}

\textcolor{blue}{
This subsection briefly explains and clarifies the computational complexity involved in \ac{CP-AFDM}.
\Eric{$N!$ is the size of the total design space for chirp-permutation operation (i.e., the total number of possible permutations), but it is not related at all to the number of objects that are stored, searched, or evaluated at each runtime (i.e., required runtime complexity).}}

\textcolor{blue}{To elaborate, the complexity of the permutation, which is an index re-ordering of a precomputed $N$-length unit sequence, is $\mathcal{O}(N)$ with zero additional complex multiplications, and even zero additional complexity if stored explicitly.
Given this, the modulator remains two diagonal multiplications with one \ac{FFT} $\mathcal{O}(N\mathrm{log}_{2}N)$, which is identical to \ac{AFDM}.}
\Eric{Therefore, for a known permutation index, the proposed \ac{CP-AFDM} has no additional computational complexity compared to the conventional \ac{AFDM} at runtime, and the only additional complexity is in the offline design and saving of the chirp-permutation index.}
%



\section{Extended Applications of Proposed \ac{CP-AFDM}}
\label{sec:application}

As clear from the previous sections, the proposed \ac{CP-AFDM} introduces a unique structural feature: the chirp-permutation domain with vast combinatorial space, which does not destroy any beneficial properties of the waveform.


In light of the above, this section highlights two exemplary application domains that leverage the flexibility in \ac{CP-AFDM} design, for enhanced performance and additional functionality in anticipation of next-generation wireless systems\Eric{\footnotemark}.

\subsection{Chirp-Permutation Index Modulation}
\label{sec:application_cpi_modulation}

\Eric{\footnotetext{\Eric{The two selected applications presented in this section are published by the authors as \cite{rou2024afdm} (chirp-permuted AFDM) and \cite{rou2025chirp} (physical-layer security framework), respectively, and we only provide a brief overview of the key concepts and results here to avoid repetition.}}}

The chirp-permutation domain inherent to the proposed \acf{CP-AFDM} waveform presents a vast combinatorial space of cardinality $N!$, all of which retain the same effective delay-Doppler channel structure and the beneficial properties, yet with vastly different channel coefficient arrangements.
Motivated by this observation, we proposed a novel waveform-embedded modulation scheme, termed \ac{CPIM} \cite{rou2024afdm}.

In the proposed scheme, additional information is encoded onto the signal by selecting a specific permuted \ac{DAFT} matrix $\mathbf{A}_{k} \in \mathbb{C}^{N \times N}$, where $k \in \{1,\ldots, K\}$ is the permutation key index determining the permutation of the one-sided permuted \ac{DAFT} matrix of \eqref{eq:cpdaft_onesided_matrix}, from a given codebook of $K$ unique permuted \ac{DAFT} matrices, used to modulate the signals.
\textcolor{blue}{Note that a codebook of $K \ll N! $ permutations is stored in a lookup table, which is designed offline, including the quantum-accelerated search for large instances~\cite{rou2024afdm}.}

Trivially, the cardinality of the codebook $K$ can be set to any value between $K = 2$ and maximally $K^{\mathrm{max}} = 2^{\lfloor\log_2(N!)\rfloor}$, where the exponent, or directly $\log_2(K)$, represents the number of bits that can be additionally encoded via the \ac{IM}, on top of the information conveyed by the modulated symbols.
This becomes immensely large even for moderate number of subcarriers $N$, i.e., $N~=~64 \longrightarrow \log_2(K^\mathrm{max}) = 295$, \Eric{which upper-bounds the additional rate} of the proposed method.

Therefore, compared to the conventional \ac{AFDM} whose rate is given by 
\begin{equation}
R^\mathrm{AFDM} \triangleq N\log_2(M) ~~\mathrm{[bits/block]},
\end{equation}
where $M$ is the modulation order (cardinality of symbol constellation), the proposed \ac{CPIM}-\ac{AFDM} scheme, without any additional transmit energy or sparsity in subcarriers and symbols (typical of conventional \ac{IM} \cite{Rou_TWC22, abu2009subcarrier}), achieves an increased rate of
\begin{equation}
R^{\mathrm{CP}\text{-}\mathrm{AFDM}} \triangleq N\log_2(M) + \log_2(K) ~~\mathrm{[bits/block]},
\end{equation}
where $K \in \big\{2,\ldots,2^{\lfloor\log_2(N!)\rfloor}\big\}$ is the number of unique chirp-permutation matrices in the codebook $\mathcal{P}$. 

It should be noted that while the transmitter-side complexity and implementation of the \ac{CPIM}-\ac{AFDM} is minimal, as it only requires a lookup table for the codebook, the corresponding detection problem is significantly more complex than conventional \ac{AFDM} or \ac{OFDM} detection, as it requires the receiver to identify the correct permutation index $k$ from the received signal, which becomes infeasible for exhaustive \ac{ML} methods to detect \ac{IM} schemes for increasing system sizes. 
To alleviate this challenge, reduced-complexity detection strategies have been proposed in \cite{rou2024afdm}, including an \ac{MMSE}-aided reduced-\ac{ML} detector and a quantum-accelerated search method based on the Grover adaptive search algorithm \cite{gilliam2021grover}, with the latter recently attracting increasing interest for large-scale \ac{IM} detection problems~\cite{Yukiyoshi_GSM}.

Beyond the joint detection of data symbols and indices, an equally important design consideration lies in the selection of the permutation codebook $\mathcal{P}$. 
This problem is non-trivial due to the combinatorial nature of the search space, with complexity growing factorially with $N$. 
Accordingly, the codebook design can be formulated as a discrete optimization problem \cite{rou2024afdm}, specifically, a dispersion maximization task in the signal space, which aims to enhance distinguishability between codewords, which can be addressed using classical combinatorial techniques and necessarily quantum optimization algorithms especially for larger system sizes~\cite{Yukiyoshi_TEQ}.

\vspace{-2ex}

\subsection{Chirp-Permutation-based Physical Layer Security}
\label{sec:application_secure}


This subsection introduces a physical-layer security scheme based on the proposed \ac{CP}-\ac{AFDM} waveform and its inherent permutation domain, as presented in~\cite{rou2025chirp}. 
The approach enables near-perfect confidentiality against eavesdroppers, even under conditions of perfect \ac{CSI} and physical colocation, without requiring any additional transmit energy, hardware complexity, or signalling overhead.

Namely, the proposed scheme considers the unique permutation index $k \in \{1, \cdots, N!\}$ used by the transmitter in the \ac{CP}-\ac{IDAFT}, as a shared secret key between legitimate transmitter and receiver.
When the legitimate receiver with the known secret key\footnote{%
\textcolor{blue}{$i_{2}$ can be expanded from a shared key via a pseudo-random generator such that the signalling requirements reduce to the standard key distribution.}} $k \in \{1, \cdots, N!\}$ detects the signal, it can use the \textit{matched} \ac{CP}-\ac{DAFT}, which uses the same chirp-permutation as the transmitter, for demodulation.


However, an eavesdropper lacking knowledge of the correct permutation index $k$ must resort to blind guessing over the full set of $N!$ valid candidates. 
If a \textit{mismatched} demodulator is employed, i.e., a \ac{CP}-\ac{DAFT} configured with an incorrect index $k' \in \{1, \cdots, N!\}$ where $k' \neq k$,  it has been shown in~\cite{rou2025chirp} that the received signal contains no extractable information, resulting in a flat \ac{BER} curve over all \ac{SNR} levels \Eric{\cite[Fig.~3]{rou2025chirp}}, matching the performance of a blind guessing approach.
%

Consequently, the only way for the eavesdropper to accurately detect the symbols is to find the secret permutation index $k$ from an exponentially large space. 
Indeed, a detailed analysis in \Eric{\cite[Fig.~2]{rou2025chirp}} confirms that both exact and near-correct guesses lead to infeasibly low success probabilities, of orders of magnitude that are near zero.
\Eric{It is further shown in \cite[Sec.~IV-A]{rou2025chirp} that} even under a brute-force search strategy, the associated computational complexity renders the attack infeasible, including when accelerated by quantum computers, where the required quantum resources exceed projected technological capabilities for the foreseeable future.
\Eric{For instance, a Grover adaptive search over all permutations of a system with $N = 3300$ subcarriers, the maximum supported in 5G~NR, requires at least $\log_2(3300!) \approx 3.4 \times 10^{4}$ logical qubits, i.e., on the order of $10^{7}$ to $10^{8}$ physical qubits after error correction, well beyond the $10^{6}$ qubits expected by 2040~\cite{rou2025chirp}.}
\Eric{Furthermore, the same principle has been applied to \ac{AFBM}, an extension of \ac{AFDM} to filter banks with low \ac{PAPR} and \ac{OOBE}~\cite{ranasinghe2025affine}, whereby permuting the two of its \ac{DAFT} chirps further enlarges the key space for greater resilience~\cite{marodin2026physical}.}

In all, compared to conventional physical-layer security techniques that often require \ac{CSI}-dependent precoding, or additional hardware and energy-consuming artificial noise generation, the proposed permutation-domain approach offers a lightweight and hardware-friendly alternative.

\vspace{-2ex}
\section{Conclusion}
\label{sec:conclusion}

\Eric{We proposed \ac{CP-AFDM}, a novel waveform that extends \ac{AFDM} with a flexible chirp-permutation domain. 
Analyses confirm that \ac{CP-AFDM} preserves the key characteristics of \ac{AFDM} for every permutation, while attaining improved \Eric{delay- and} Doppler-domain \ac{PSLR} under deterministic probes. 
In addition, the permutation domain further enables high-rate index modulation and lightweight, practically perfect physical-layer security, positioning \ac{CP-AFDM} as a strong candidate for next-generation systems supporting high-mobility communication, \ac{ISAC}, and secure transmission.}

\vspace{-1ex}




\vspace{-1ex}
\def\IEEEbibitemsep{0pt}
{\renewcommand{\footnotesize}{\fontsize{7.72}{8.175}\selectfont}
\linespread{1}\selectfont
\bibliographystyle{IEEEtran}
\bibliography{bstcontrol,ref_abbr}

@IEEEtranBSTCTL{IEEEexample:BSTcontrol,
  CTLuse_forced_etal       = "yes",
  CTLmax_names_forced_etal = "1",
  CTLnames_show_etal       = "1",
  CTLuse_alt_spacing       = "yes",
  CTLalt_stretch_factor    = "4",
  CTLdash_repeated_names   = "yes",
  CTLuse_url               = "yes"
}

@article{saad2019vision,
  title={A vision of {6G} wireless systems: Applications, trends, technologies, and open research problems},
  author={Saad, Walid and Bennis, Mehdi and Chen, Mingzhe},
  journal={IEEE Netw.},
  volume={34},
  number={3},
  pages={134--142},
  year={2019},
  publisher={IEEE}
}

@article{rou2025affine,
  title={{Affine Frequency Division Multiplexing (AFDM) for 6G: Properties, Features, and Challenges}},
  author={Rou, Hyeon Seok and Ranasinghe, Kuranage Roche Rayan and Savaux, Vincent and de Abreu, Giuseppe Thadeu Freitas and Gonz{\'a}lez G., David and Masouros, Christos},
  journal={IEEE Commun. Stand. Mag.},
  year={2026},
  volume={10},
  number={2},
  pages={216--225}
}

@article{ranasinghe2025affine,
  title={{Affine Filter Bank Modulation (AFBM): A Novel 6G ISAC Waveform With Low PAPR and OOBE}},
  author={Ranasinghe, Kuranage Roche Rayan and Senger, Henrique L and Gon{\c{c}}alves, Gustavo P and Rou, Hyeon Seok and Chang, Bruno S and de Abreu, Giuseppe Thadeu Freitas and Le Ruyet, Didier},
  journal={IEEE Trans. Wireless Commun.},
  year={2026},
  volume={25},
  pages={12754--12769}
}

@ARTICLE{Wang_6G,
author={Wang, Cheng-Xiang and You, Xiaohu and Gao, Xiqi and Zhu, Xiuming and Li, Zixin and Zhang, Chuan and Wang, Haiming and Huang, Yongming and Chen, Yunfei and Haas, Harald and Thompson, John S. and Larsson, Erik G. and Renzo, Marco Di and Tong, Wen and Zhu, Peiying and Shen, Xuemin and Poor, H. Vincent and Hanzo, Lajos},
  journal={IEEE Commun. Surveys Tuts.}, 
title={{On the Road to {6G}: Visions, Requirements, Key Technologies, and Testbeds}}, 
year={2023},
volume={25},
number={2},
pages={905-974},
}

@article{agiwal2016next,
  title={Next generation {5G} wireless networks: A comprehensive survey},
  author={Agiwal, Mamta and Roy, Abhishek and Saxena, Navrati},
  journal={IEEE Commun. Surveys Tuts.},
  volume={18},
  number={3},
  pages={1617--1655},
  year={2016},
  publisher={IEEE}
}

@book{Bliss_DDchannel,
title={{Dispersive and doubly dispersive channels}},
author={Bliss, DW and Govindasamy, S},
year={2013},
publisher={Cambridge University Press}
}

@ARTICLE{Rou_SPM24,
author={Rou, Hyeon Seok and de Abreu, Giuseppe Thadeu Freitas and Choi, Junil and Gonz{\'a}lez G., David and Kountouris, Marios and Guan, Yong Liang and Gonsa, Osvaldo},
  journal={IEEE Signal Process. Mag.}, 
title={{From Orthogonal Time-Frequency Space to Affine Frequency-Division Multiplexing: A comparative study of next-generation waveforms for integrated sensing and communications in doubly dispersive channels}}, 
year={2024},
volume={41},
number={5},
pages={71-86},
}

@article{armada2001understanding,
  title={Understanding the effects of phase noise in orthogonal frequency division multiplexing ({OFDM})},
  author={Armada, A Garcia},
  journal={IEEE Trans. Broadcast.},
  volume={47},
  number={2},
  pages={153--159},
  year={2001},
  publisher={IEEE}
}

@article{liyanaarachchi2021optimized,
  title={Optimized waveforms for {5G}--{6G} communication with sensing: Theory, simulations and experiments},
  author={Liyanaarachchi, Sahan Damith and Riihonen, Taneli and Barneto, Carlos Baquero and Valkama, Mikko},
  journal={IEEE Trans. Wireless Commun.},
  volume={20},
  number={12},
  pages={8301--8315},
  year={2021},
  publisher={IEEE}
}

@article{hadani2018otfs,
  title={{OTFS}: A new generation of modulation addressing the challenges of {5G}},
  author={Hadani, Ronny and Monk, Anton},
  journal={arXiv:1802.02623},
  year={2018}
}

@article{jian2024fda,
  title={{FDA-MIMO}-based integrated sensing and communication system with frequency offsets permutation index modulation},
  author={Jian, Jiangwei and Huang, Qimao and Huang, Bang and Wang, Wen-Qin},
  journal={IEEE Trans. Commun.},
  volume={72},
  number={11},
  pages={6707--6721},
  year={2024},
  publisher={IEEE}
}

@inproceedings{niu2025afdm,
  title={{AFDM}-based integrated system for underwater detection and communication waveform design},
  author={Niu, Qixiang and Shi, Wentao and Ram{\'\i}rez, David and Jing, Lianyou and Zhang, Qunfei},
  booktitle={Proc. IEEE WCNC},
  pages={1--6},
  year={2025},
}

@article{wen2017multiple,
  title={Multiple-mode orthogonal frequency division multiplexing with index modulation},
  author={Wen, Miaowen and Basar, Ertugrul and Li, Qiang and Zheng, Beixiong and Zhang, Meng},
  journal={IEEE Trans. Commun.},
  volume={65},
  number={9},
  pages={3892--3906},
  year={2017},
  publisher={IEEE}
}

@ARTICLE{liu2025cp,
  author={Liu, Fan and Zhang, Ying and Xiong, Yifeng and Li, Shuangyang and Yuan, Weijie and Gao, Feifei and Jin, Shi and Caire, Giuseppe},
  journal={IEEE Trans. Inf. Theory}, 
  title={{CP-OFDM} Achieves the Lowest Average Ranging Sidelobe Under {QAM/PSK} Constellations}, 
  year={2025},
  volume={71},
  number={9},
  month={Sep.},
  pages={6950-6967},
}

@inproceedings{rou2026normalized,
  title={Normalized ambiguity function characteristics of {OFDM, OTFS, AFDM, and CP-AFDM for ISAC}},
  author={Rou, Hyeon Seok and de Abreu, Giuseppe Thadeu Freitas},
  booktitle={IEEE ICC},
  pages={1--6},
  year={2026},
}

@ARTICLE{yu2022location,
  author={Yu, Zhouyuan and Hu, Xiaoling and Liu, Chenxi and Peng, Mugen and Zhong, Caijun},
  journal={IEEE Trans. Signal Process.}, 
  title={Location Sensing and Beamforming Design for IRS-Enabled Multi-User {ISAC} Systems}, 
  year={2022},
  volume={70},
  number={},
  month={Nov.},
  pages={5178-5193},
}

@ARTICLE{wang2024ris,
  author={Wang, Zhen and Hu, Xiaoling and Liu, Chenxi and Peng, Mugen},
  journal={IEEE Trans. Wireless Commun.}, 
  title={{RIS}-Enabled Multi-Target Sensing: Performance Analysis and Space-Time Beamforming Design}, 
  year={2024},
  volume={23},
  number={10},
  month={Jun.},
  pages={13889-13903},
}

@article{gopalam2024zak,
  title={{Zak-OTFS} implementation via time and frequency windowing},
  author={Gopalam, Swaroop and Collings, Iain B and Hanly, Stephen V and Inaltekin, Hazer and Pillai, Sibi Raj B and Whiting, Philip},
  journal={IEEE Trans. Commun.},
  volume={72},
  number={7},
  pages={3873--3889},
  year={2024},
  publisher={IEEE}
}

@article{ranasinghe2024joint,
  title={Joint channel, data and radar parameter estimation for {AFDM} systems in doubly-dispersive channels},
  author={Ranasinghe, Kuranage Roche Rayan and Rou, Hyeon Seok and de Abreu, Giuseppe Thadeu Freitas and Takahashi, Takumi and Ito, Kenta},
  journal={IEEE Trans. Wireless Commun.},
  year={2025},
  volume={24},
  number={2},
  pages={1602-1619},
  publisher={IEEE}
}

@article{luo2025novel,
  title={A Novel Angle-Delay-Doppler Estimation Scheme for {AFDM-ISAC} System in Mixed Near-field and Far-field Scenarios},
  author={Luo, Yirui and Guan, Yong Liang and Ge, Yao and Gonz{\'a}lez, David and Yuen, Chau},
  journal={IEEE Internet Things J.},
  year={2025},
  publisher={IEEE},
  volume={12},
  number={13},
  pages={22669--22682}
}

@article{zhu2024afdm,
  title={{AFDM}-based bistatic integrated sensing and communication in static scatterer environments},
  author={Zhu, Jiajun and Tang, Yanqun and Liu, Fan and Zhang, Xiaoying and Yin, Haoran and Zhou, Yu},
  journal={IEEE Wireless Commun. Lett.},
  volume={13},
  number={8},
  pages={2245--2249},
  year={2024},
  publisher={IEEE}
}

@article{surabhi2019peak,
  title={Peak-to-average power ratio of {OTFS} modulation},
  author={Surabhi, GD and Augustine, Rose Mary and Chockalingam, Ananthanarayanan},
  journal={IEEE Commun. Lett.},
  volume={23},
  number={6},
  pages={999--1002},
  year={2019},
  publisher={IEEE}
}

@article{tan2025two,
  title={Two-step windowing for both sidelobe suppression and {SNR} maximization of radar range profiles},
  author={Tan, Donglin and Wang, Junfeng},
  journal={Signal Image Video Process.},
  volume={19},
  number={2},
  pages={160},
  year={2025},
  publisher={Springer}
}

@article{azouz2021design,
  title={Design and implementation of an enhanced matched filter for sidelobe reduction of pulsed linear frequency modulation radar},
  author={Azouz, Ahmed and Abosekeen, Ashraf and Nassar, Sameh and Hanafy, Mohamed},
  journal={Sensors},
  volume={21},
  number={11},
  pages={3835},
  year={2021},
  publisher={MDPI}
}

@inproceedings{luo2025target,
  title={Target sensing with off-grid sparse bayesian learning for {AFDM-ISAC} system},
  author={Luo, Yirui and Guan, Yong Liang and Ge, Yao and Yuen, Chau},
  year={2025},
  booktitle={Proc. IEEE ICC Workshops},
  pages={881--886}
}

@inproceedings{ranasinghe2024fast,
  title={Fast and efficient sequential radar parameter estimation in {MIMO-OTFS} systems},
  author={Ranasinghe, Kuranage Roche Rayan and Rou, Hyeon Seok and de Abreu, Giuseppe Thadeu Freitas},
  booktitle={Proc. IEEE ICASSP},
  pages={8661--8665},
  year={2024},
}

@INPROCEEDINGS{Bemani_AFDM21,
  author={Bemani, Ali and Ksairi, Nassar and Kountouris, Marios},
  booktitle={IEEE ICC Workshops}, 
  title={{AFDM}: A Full Diversity Next Generation Waveform for High Mobility Communications}, 
  year={2021},
  volume={},
  number={},
  pages={1-6},
}

@ARTICLE{Bemani_AFDM23,
  author={Bemani, Ali and Ksairi, Nassar and Kountouris, Marios},
  journal={IEEE Trans. Wireless Commun.}, 
  title={Affine Frequency Division Multiplexing for Next Generation Wireless Communications}, 
  year={2023},
  volume={22},
  number={11},
  pages={8214-8229},
}

@article{ouyang2016orthogonal,
  title={Orthogonal chirp division multiplexing},
  author={Ouyang, Xing and Zhao, Jian},
  journal={IEEE Trans. Commun.},
  volume={64},
  number={9},
  pages={3946--3957},
  year={2016},
  publisher={IEEE}
}

@article{gaudio2020effectiveness,
  title={On the effectiveness of {OTFS} for joint radar parameter estimation and communication},
  author={Gaudio, Lorenzo and Kobayashi, Mari and Caire, Giuseppe and Colavolpe, Giulio},
  journal={IEEE Trans. Wireless Commun.},
  volume={19},
  number={9},
  pages={5951--5965},
  year={2020},
  publisher={IEEE}
}

@INPROCEEDINGS{RanasingheWCNC,
  author={Ranasinghe, Kuranage Roche Rayan and Ando, Kengo and Rou, Hyeon Seok and de Abreu, Giuseppe Thadeu Freitas and Bathelt, Andreas},
  booktitle={Proc. IEEE WCNC}, 
  title={Blind Bistatic Radar Parameter Estimation in Doubly-Dispersive Channels}, 
  year={2025},
  volume={},
  number={},
  pages={1-6},
}

@article{rou2026evolving,
  title={{AFDM}: Evolving {OFDM} Toward {6G+}},
  author={Rou, Hyeon Seok and Savaux, Vincent and Sui, Zeping and de Abreu, Giuseppe Thadeu Freitas and Liu, Zilong},
  journal={IEEE Open J. Commun. Soc.},
  year={2026},
  volume={7},
  pages={12408--12434}
}

@article{mirabella2026continuous,
  title={Continuous-time analysis of {AFDM}: {Pulse}-shaping, fundamental bounds, and impact of hardware impairments},
  author={Mirabella, Michele and Rou, Hyeon Seok and Di Viesti, Pasquale and de Abreu, Giuseppe Thadeu Freitas and Vitetta, Giorgio Matteo},
  journal={IEEE J. Sel. Areas Commun.},
  volume={44},
  pages={6390--6406},
  year={2026}
}

@article{zhu2023design,
  title={Design and performance analysis of index modulation empowered {AFDM} system},
  author={Zhu, Jing and Luo, Qu and Chen, Gaojie and Xiao, Pei and Xiao, Lixia},
  journal={IEEE Wireless Commun. Lett.},
  volume={13},
  number={3},
  pages={686--690},
  year={2023},
  publisher={IEEE}
}

@inproceedings{liu2024pre,
  title={Pre-chirp-domain index modulation for affine frequency division multiplexing},
  author={Liu, Guangyao and Mao, Tianqi and Liu, Ruiqi and Xiao, Zhenyu},
  booktitle={Proc. IWCMC},
  pages={0473--0478},
  year={2024},
}

@article{yin2025ambiguity,
  title={Ambiguity Function Analysis of {AFDM} Signals for Integrated Sensing and Communications},
  author={Yin, Haoran and Tang, Yanqun and Ni, Yuanhan and Wang, Zulin and Chen, Gaojie and Xiong, Jun and Yang, Kai and Kountouris, Marios and Guan, Yong Liang and Zeng, Yong},
  journal={IEEE J. Sel. Areas Commun.},
  year={2026},
  volume={44},
  pages={196--211}
}

@inproceedings{bedeer2025ambiguity,
  title={Ambiguity Function Analysis of Affine Frequency Division Multiplexing for {ISAC}},
  author={Bedeer, Ebrahim},
  year={2025},
  booktitle={Proc. IEEE GLOBECOM},
  pages={2850--2855}
}

@article{knagenhjelm2002hadamard,
  title={The {Hadamard} transform---a tool for index assignment},
  author={Knagenhjelm, Petter and Agrell, Erik},
  journal={IEEE Trans. Inf. Theory},
  volume={42},
  number={4},
  pages={1139--1151},
  year={1996},
  publisher={IEEE}
}

@article{ahmed2006discrete,
  title={Discrete cosine transform},
  author={Ahmed, Nasir and Natarajan, T and Rao, Kamisetty R},
  journal={IEEE Trans. Comput.},
  volume={C-23},
  number={1},
  pages={90--93},
  year={1974},
  publisher={IEEE}
}

@article{surabhi2019diversity,
  title={On the diversity of uncoded {OTFS} modulation in doubly-dispersive channels},
  author={Surabhi, GD and Augustine, Rose Mary and Chockalingam, A},
  journal={IEEE Trans. Wireless Commun.},
  volume={18},
  number={6},
  pages={3049--3063},
  year={2019},
  publisher={IEEE}
}

@article{muppaneni2023channel,
  title={Channel and radar parameter estimation with fractional delay-Doppler using {OTFS}},
  author={Muppaneni, Sai Pradeep and Mattu, Sandesh Rao and Chockalingam, A},
  journal={IEEE Commun. Lett.},
  volume={27},
  number={5},
  pages={1392--1396},
  year={2023},
  publisher={IEEE}
}

@article{savaux2024special,
  title={Special cases of {DFT}-based modulation and demodulation for affine frequency division multiplexing},
  author={Savaux, Vincent},
  journal={IEEE Trans. Commun.},
  volume={72},
  number={12},
  pages={7627--7638},
  year={2024},
  publisher={IEEE}
}

@article{yin2024diagonally,
  title={Diagonally reconstructed channel estimation for {MIMO-AFDM} with inter-doppler interference in doubly selective channels},
  author={Yin, Haoran and Wei, Xizhang and Tang, Yanqun and Yang, Kai},
  journal={IEEE Trans. Wireless Commun.},
  volume={23},
  number={10},
  pages={14066--14079},
  year={2024},
  publisher={IEEE}
}

@inproceedings{rou2024afdm,
  title={{AFDM} chirp-permutation-index modulation with quantum-accelerated codebook design},
  author={Rou, Hyeon Seok and Yukiyoshi, Kein and Mikuriya, Taku and de Abreu, Giuseppe Thadeu Freitas and Ishikawa, Naoki},
  booktitle={Proc. Asilomar Conf. Signals, Syst., Comput.},
  pages={817--821},
  year={2024},
}

@article{rou2025chirp,
  title={Chirp-Permuted {AFDM} for Quantum-Resilient Physical-Layer Secure Communications},
  author={Rou, Hyeon Seok and de Abreu, Giuseppe Thadeu Freitas},
  journal={IEEE Wireless Commun. Lett.},
  year={2025},
  publisher={IEEE},
  volume={14},
  number={8},
  pages={2376--2380}
}

@ARTICLE{Rou_TWC22,
  author={Rou, Hyeon Seok and de Abreu, Giuseppe Thadeu Freitas and Iimori, Hiroki and G., David González and Gonsa, Osvaldo},
  journal={IEEE Trans. Wireless Commun.}, 
  title={Scalable Quadrature Spatial Modulation}, 
  year={2022},
  volume={21},
  number={11},
  pages={9293-9311},
}

@article{gilliam2021grover,
  title={Grover adaptive search for constrained polynomial binary optimization},
  author={Gilliam, Austin and Woerner, Stefan and Gonciulea, Constantin},
  journal={Quantum},
  volume={5},
  pages={428},
  year={2021},
  publisher={Verein zur F{\"o}rderung des Open Access Publizierens in den Quantenwissenschaften}
}

@ARTICLE{Yukiyoshi_TEQ,
  author={Yukiyoshi, Kein and Mikuriya, Taku and Rou, Hyeon Seok and de Abreu, Giuseppe Thadeu Freitas and Ishikawa, Naoki},
  journal={IEEE Trans. Quantum Eng.}, 
  title={Quantum Speedup of the Dispersion and Codebook Design Problems}, 
  year={2024},
  volume={5},
  number={},
  pages={1-16},
}

@INPROCEEDINGS{Yukiyoshi_GSM,
  author={Yukiyoshi, Kein and Mikuriya, Taku and Rou, Hyeon Seok and de Abreu, Giuseppe Thadeu Freitas and Ishikawa, Naoki},
  booktitle={IEEE VTC-Fall}, 
  title={Grover Adaptive Search for Maximum Likelihood Detection of Generalized Spatial Modulation}, 
  year={2024},
  volume={},
  number={},
  pages={1-5},
}

@inproceedings{abu2009subcarrier,
  title={Subcarrier-index modulation {OFDM}},
  author={Abu-Alhiga, Rami and Haas, Harald},
  booktitle={Proc. IEEE PIMRC},
  pages={177--181},
  year={2009},
}

@article{sui2026achieving,
  title={Achieving Secure and Low-{PAPR} {AFDM} Transmission with Pre-Chirp Signal Design},
  author={Sui, Zeping and Liu, Zilong and Savaux, Vincent and Yang, Lie-Liang and Musavian, Leila},
  journal={Preprints, doi: 10.20944/preprints202611.1066.v1},
  year={2026}
}

@article{ali2025spreading,
  title={Spreading the Wave: Low-Complexity {PAPR} Reduction for {AFDM and OCDM in 6G} Networks},
  author={Ali, Afan and Arous, Abdelali and Arslan, Huseyin},
  journal={IEEE Trans. Green Commun. Netw.},
  year={2026},
  volume={10},
  pages={1565--1577}
}

@article{sui2026generalized,
  title={Generalized Spatial Modulation Aided Affine Frequency Division Multiplexing},
  author={Sui, Zeping and Liu, Zilong and Musavian, Leila and Yang, Lie-Liang and Hanzo, Lajos},
  journal={IEEE Trans. Wireless Commun.},
  volume={25},
  pages={4658--4673},
  year={2026},
}

@article{rou2026resurrection,
  title={The resurrection of spectrum spreading for {6G} and beyond: {From} sinusoids to chirps},
  author={Rou, Hyeon Seok and de Abreu, Giuseppe Thadeu Freitas and Bj{\"o}rnson, Emil and Kim, Sunwoo and Kountouris, Marios},
  journal={IEEE Wireless Commun., in press},
  year={2026}
}

@article{marodin2026physical,
  title={Physical Layer Security in Chirp Permuted {AFBM}},
  author={Marodin, Pedro Schuves and Senger, Henrique L. and Gon{\c{c}}alves, Gustavo P. and Rou, Hyeon Seok and Chang, Bruno S. and de Abreu, Giuseppe Thadeu Freitas and Le Ruyet, Didier},
  journal={arXiv:2610.33276},
  year={2026}
}
}

\end{document}